%% file: main.tex
\documentclass[a4paper,twocolumn,unpublished]{quantumarticle}
\pdfoutput=1

\include{tikzit.sty}
\input{tikzit.tikzstyles}

\usepackage{tikz}
\usepackage{subcaption}
\usepackage{bookmark}
\usepackage{amssymb}
\usepackage{enumerate}
\usepackage{listings}
\usepackage{multirow}
\usepackage{tikz}
\usepackage{amsmath}
\usepackage{amsthm}
\usepackage{amssymb}
\usepackage{float}
\usepackage{bbm}
\usepackage{thm-restate}

\usepackage[labelfont={bf},font={small,it}]{caption}

\newtheorem{theorem}{Theorem}
\newtheorem{lemma}{Lemma}
\newtheorem{algorithm}{Algorithm}
\newtheorem{proposition}{Proposition}

\newtheorem{remark}{Remark}
\theoremstyle{definition}
\newtheorem{definition}{Definition}
\usepackage[numbers]{natbib}

\newcommand{\tx}{\textrm}

\newcommand{\M}{\mathcal}
\newcommand{\ceil}[1]{\left\lceil#1\right\rceil}
\newcommand{\floor}[1]{\left\lfloor#1\right\rfloor}
\newcommand{\Mod}{\ \mathrm{mod}\ }
\newcommand{\problem}[1]{{\normalfont\texttt{#1}}}
\newcommand{\Odd}{\textrm{Odd}}

\title{Finding trail covers: near-optimal decompositions of graph states as linear fusion networks}
\author[1]{William Cashman}
\author[2]{Giovanni de Felice}
\author[1]{Aleks Kissinger}
\affil[1]{University of Oxford, United Kingdom}
\affil[2]{Quantinuum, 17 Beaumont Street, Oxford, OX1 2NA, United Kingdom}

\begin{document}

\maketitle

\begin{abstract}
Quantum compilation requires the development of new algorithms that optimise the cost of implementing quantum computations on physical hardware.
Often this gives rise to problems which are asymptotically hard to solve classically, and for which heuristics and reductions to known problems are of great practical use. 
In this paper, we study three graph-theoretic problems which can be seen as generalisations of the Eulerian and Hamiltonian path problems. 
These arise in photonic implementations of measurement-based quantum computing, where graph states are constructed by fusing bounded-length linear resource states. 
Since the fusion operation succeeds with probability smaller than one, we wish to minimise the number of fusions required to build a particular graph state
and this corresponds to finding a minimal path or trail cover of the graph.
We show that these covering problems are NP-hard in most cases and give heuristic algorithms for finding trail covers in graphs including a reduction to the travelling salesman problem.
We propose new rewrite strategies for graph states that reduce the number of fusions required to build a given graph. 
Finally, we apply these algorithms to the compilation of photonic fusion networks and provide a series of benchmarks showing the performance of our algorithms on common error-correcting codes and circuits from the QASMBench set. 
\end{abstract}

\section{Introduction}\label{sec_introduction}

In the measurement-based or ``one-way'' model \cite{raussendorf_one-way_2001} of quantum computing (MBQC),
the computation consists of preparing an entangled graph state $G$ over multiple qubits, 
and performing a sequence of adaptive measurements on the nodes of this graph.
A particular feature of MBQC is that every qubit in the resource state is only measured once and interacts only with its nearest neighbours on the graph.
This makes it particularly suitable to photonic architectures since photons are destroyed upon measurement.
In practice, the large graph $G$ must be built by entangling smaller resource states.
However, as entanglement can only be generated probabilistically in linear optics \cite{stanisic_generating_2017},
the efficient construction of large graph states poses a serious challenge to photonic architectures.

Current approaches to photonic graph state generation are based on a particular type of entangling operation, 
known as fusion measurement \cite{Browne_2005}.
Many small resource states are produced and are entangled together by a series of multi-qubit fusion operations \cite{repub_fbqc}.
We consider access to a source of photons, producing a constant-size resource state at every time step. 
Photons from different sources or time-steps are then fused together to construct a given graph $G$.
Usually, the photon source only provides states with limited entanglement --- such as star graphs, lines or small polyhedra 
--- and the fusion pattern determines long-range correlations.
The information of which qubits from different resource states are to be fused together is called \emph{fusion network}.

We consider linear fusion networks, where the resource states have a one dimensional entanglement structure.
These include n-GHZ states, linear cluster states and variations of the two \cite{cite_ft_quantum_computing_with_nondeterminism, cite_photonic_architecture_with_ghz, de_jong_extracting_2024}.
They can be generated with high probability $p_R \approx 1$ from matter-based photon sources such as ions \cite{blinov_observation_2004} and quantum dots \cite{istrati_sequential_2020,quant-dots}, 
or from multiple SPDC sources and linear optics using active multiplexing to boost $p_R$ \cite{bartolucci_creation_2021}.
Fusion measurements are a more costly operation as they succeed with probability $p_S$ smaller than $p_R$.
As a consequence, we wish to minimise the number of fusions required to build a given graph state to maximise the probability of successfully executing the MBQC pattern.

We use two types of fusion operations which can be implemented on pairs of photonic qubits using linear optics only.
The first, known as Type II fusion \cite{Browne_2005} and used in \cite{og_fbqc}, has the effect of merging two nodes of a graph state into one, as depicted in Figure~\ref{fig:X-fusion}.
We call it $X$ fusion as it corresponds to adding a qubit in the graph, connected to the target nodes, and measuring it in the Pauli $X$ basis.
In the success case, this implements a $ZZ$ measurement on the target qubits.
The second fusion operation, called $Y$ fusion and depicted in Figure~\ref{fig:Y-fusion}, is used in \cite{lim_repeat-until-success_2005, gliniasty_spin-optical_2024} to implement linear optical CZ gates probabilistically.
In the success case, it has the effect of performing a CZ gate on the target qubits. 
Up to local Clifford operations, this corresponds to adding a node in the graph and measuring it in the Pauli $Y$ basis.
\begin{figure}[h]
    \centering
    \input{diagrams/x_fusion_notation.tikz}
    \caption{An X fusion (or Type-II) corresponds to adding a node in the graph measured in the $X$ basis.}
    \label{fig:X-fusion}
\end{figure}
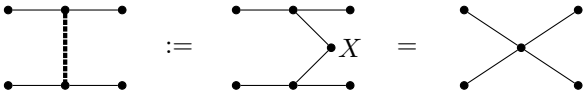
\begin{figure}[h]
    \centering
    \input{diagrams/y_fusion_notation.tikz}
    \caption{A Y fusion (or CZ) corresponds to adding a node in the graph measured in the $Y$ basis, and rotating the target qubits with a $Z$ phase of $\frac{\pi}{2}$.}
    \label{fig:Y-fusion}
\end{figure}
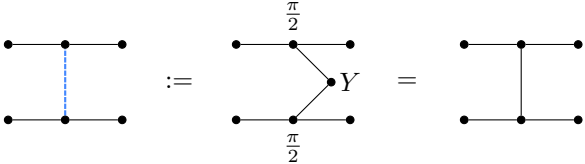

Given additional resource entanglement, both fusion operations admit protocols to increase the probability of success $p_S$ \cite{lim_repeat-until-success_2005}, 
although X fusion is sometimes preferred as its boosting protocol does not require active switching \cite{lee_nearly_2015, hilaire_enhanced_2024}.
Note that X and Y fusions cannot be related by local Clifford operations on the target qubits.
They are in fact canonical representatives of entangling stabilizer measurements that can be performed with linear optics.
Moreover, any decomposition of an MBQC graph (with flow structure) as a fusion network of $X$ and $Y$ fusions admits a deterministic measurement pattern to implement the target graph \cite{felice_fusion_2024}.

The above discussion motivates the following compilation problem: given a graph $G$ find a linear fusion network that implements the graph using the minimal number of $X$ or $Y$ fusions.
Translating this problem in the language of graph theory, we want to find a set of edge-disjoint trails on the graph, which covers every node of the graph, and minimises the number of intersections and un-covered edges. 
We call such a collection of trails a \textit{trail cover} and show in Section~\ref{sec:problem} that minimising fusions is equivalent to finding a trail cover of the graph with the fewest number of trails.

This graph-theoretic formulation of the problem allows us to relate it to other problems whose complexity is known in Section~\ref{sec:complexity}.
We consider different variations of the minimum trail cover problem. 
First by restricting the allowed fusion operations to only $X$ or only $Y$, we obtain generalisations of the Eulerian trail and Hamiltonian path problems, respectively.
We show that the case for $X$ fusions can be solved in polynomial time while the case for $Y$ fusions is NP-hard. 
We then use these results to show that the case where both $X$ and $Y$ fusions are allowed, corresponding to the minimum trail cover problem, is also NP-hard.

In Section~\ref{sec:approximation} we provide efficient algorithms for approximating the solutions of these NP-hard problems.
We give a heuristic strategy based on subdividing Eulerian trails for solving the minimum trail decomposition and trail cover problems.
We reduce the minimum trail cover problem to the Travelling Salesman Problem (TSP), allowing us to use a TSP solver for large instances of the problem.

The rewriting theory of graph states \cite{duncan_rewriting_2010, duncan_graph-theoretic_2019} allows us to consider equivalent graphs that have different topology but implement the same quantum state. 
We are particularly interested in rewrites, such as local complementation, that preserve \emph{flow structure}\cite{browne_generalized_2007,backens_there_2021,pauli_preserving_rewrites} on the graph, to ensure that the resulting MBQC computation is deterministic.
For the graph state generation problem, this means that it is sufficient to implement any graph $G'$ such that we can rewrite $G \to G'$.
In Section~\ref{sec:graphrewrites}, we thus propose new rewrite strategies to reduce the number of trails in minimum trail covers.

Finally, in Section~\ref{sec:benchmarks}, we evaluate the performance of the developed graph rewrites and trail-cover finding algorithms 
in a dataset of common error-correcting codes \cite{cite_qec_benchmarks} and circuits from the QASM benchmark set~\cite{cite_qasm_bench}.
We prove a theoretical lower bound to the number of fusions required to decompose a graph state into linear resource states with a bounded number of photons.
By allowing both X and Y fusions, our algorithms achieve results remarkably close to the lower bound, 
offering improvements even in small graphs such as the Shor-(2, 2) encoded $6$-ring studied in \cite{og_fbqc,wein_minimizing_2024}.
Our approach extends previous work on graph state generation \cite{cite_qec_benchmarks, cite_vienna_team} with resource states of arbitrary length, different fusion types, and new rewriting heuristics.
We moreover address the probabilistic aspects of fusion measurements by optimising the number of fusion attempts 
in repeat-until-success schemes with photon-bounded resource states \cite{cite_rus_original,lee_nearly_2015, gliniasty_spin-optical_2024},
obtaining resource estimates for near-deterministic implementations of MBQC graphs.

\section{Problem formulation}\label{sec:problem}
\input{1-problem}

\section{Complexity analysis}\label{sec:complexity}
\input{2-complexity}

\section{Efficient approximation algorithms}\label{sec:approximation}
\input{3-approximation}

\section{Heuristic graph rewrites}\label{sec:graphrewrites}
\input{5-graphrewrites}

\section{Compilation benchmarks}\label{sec:benchmarks}
\input{6-benchmarks}

\section{Conclusion}

Graph state generation from resource states and fusion measurements is a key component of photonic implementations of measurement-based quantum computing \cite{cite_orig_fbqc}.
Because of the probabilistic nature of linear optical entanglement, it is essential to minimise the number of fusion measurements required to build a given graph state.
In this work, we proposed a formalisation of this compilation problem in terms of trail covers where the objective minimizes both the number of fusions and the number of linear resources required to build the graph, see Theorem~\ref{thm_num_fusions}.
We then gave complexity-theoretic reductions showing that this problem is NP-hard in most cases, Section~\ref{sec:complexity}. 
In practice, however, it is possible to find approximate solutions to these problems for which we developed efficient algorithms and heuristic strategies in Section~\ref{sec:approximation}.
Leveraging the rich theory of graph states \cite{duncan_graph-theoretic_2019,backens_there_2021}, in Section~\ref{sec:graphrewrites}, we further proposed graph rewriting strategies that reduce the number of fusions required to build a given quantum state.
Our benchmarks, in Section~\ref{sec:benchmarks}, analyse the number of fusions, photons and resource states required to build common error correcting codes from \cite{cite_qec_benchmarks,og_fbqc} and circuits from the QASM benchmark \cite{cite_qasm_bench}. 
The results show that (i) using different fusion types offers substantial improvements in all metrics,
(ii) the XY trail covers produced by our algorithm closely approach a theoretical lower bound (especially in large sparse graphs),
and (iii) our rewriting heuristics allow for further reductions in the number of fusions and photons, at the cost of a slight increase in the number of resource states.

Moreover, we establish error bounds on the algorithms for finding trail covers and trail decompositions. For any graph state with corresponding graph $G$, we are able to construct an $X$ fusion network with resource states of length $L$ with no more than $\frac{1}{2}|\Odd(G)|(1 - \frac{3}{L}) + 1$ fusions that the minimum (Proposition~\ref{prop_photon_trail_approx}).
% Ideally we should use the value of the photon bounded case
For $XY$ fusion networks our algorithm first constructs an X fusion network and uses a heuristic to remove resource states and replace them with $Y$ fusions to reduce the total number of fusions. Our heuristic is not guaranteed to find such trail and so it may produce up to $\frac{1}{2}|\Odd(G)| - 1$ more fusions than the minimum. However, the benchmarks have shown the heuristic to be highly effective on real-world algorithms and closely approximating the minimum

Our framework for finding trail covers unifies previous compilation work in fusion-based architectures \cite{cite_qec_benchmarks, cite_vienna_team,wein_minimizing_2024}, 
offering more flexibility in the length of resource states, the fusion type, and the rewriting strategy (using Z-insertions/deletions in addition to local complementation).
Note that our approach separates the fusion and local complementation stages in the protocol.
As shown for small graphs in \cite{lobl_generating_2025}, intermediary local complementations can further minimise the number of fusions required.
A generalisation of our notion of trail cover would be needed to capture the setting with local complementations and optimise 
the full compilation problem on larger graphs, although we expect that the asymptotic complexity of the problem will remain $NP$-hard.

We analysed the probability of successful graph state generation in a repeat-until-success fusion-based architecture \cite{lee_nearly_2015, gliniasty_spin-optical_2024,hilaire_enhanced_2024}, 
indicating resource cost lower bounds for the implementation of graph states in the presence of photon loss.
We only analysed the probability of all fusions being successful, and its boosting via the RUS protocol. 
In practice however, some fusion failures and photon losses could be tolerated if the logical quantum state is redundantly encoded in a larger graph state \cite{bell_optimizing_2023, lobl_loss-tolerant_2024}.
We leave the analysis of loss-tolerance to future work, noting that different trail covers or graph rewrites may lead to different loss-tolerant properties on the encoded graph states.

\bibliographystyle{quantum}
\bibliography{references}

\appendix
\input{appendix.tex}

\end{document}

%% file: tikzit.tikzstyles
\tikzstyle{gate}=[shape=rectangle, text height=1.5ex, text depth=0.25ex, yshift=0.5mm, fill=white, draw=black, minimum height=5mm, yshift=-0.5mm, minimum width=5mm, font={\small}, tikzit category=circuit]
\tikzstyle{Z dot}=[inner sep=0mm, minimum size=2mm, shape=circle, draw=black, fill={rgb,255: red,221; green,255; blue,221}, tikzit category=zx]
\tikzstyle{Z phase dot}=[minimum size=5mm, font={\footnotesize\boldmath}, shape=rectangle, rounded corners=2mm, inner sep=0.2mm, outer sep=-2mm, scale=0.8, tikzit shape=circle, draw=black, fill={rgb,255: red,221; green,255; blue,221}, tikzit draw=blue, tikzit category=zx]
\tikzstyle{X dot}=[Z dot, shape=circle, draw=black, fill={rgb,255: red,255; green,136; blue,136}, tikzit category=zx]
\tikzstyle{X phase dot}=[Z phase dot, tikzit shape=circle, tikzit draw=blue, fill={rgb,255: red,255; green,136; blue,136}, font={\footnotesize\boldmath}, tikzit category=zx]
\tikzstyle{hadamard}=[fill=yellow, draw=black, shape=rectangle, inner sep=0.6mm, minimum height=1.5mm, minimum width=1.5mm, tikzit category=zx]
\tikzstyle{vertex}=[inner sep=0mm, minimum size=1mm, shape=circle, draw=black, fill=black, tikzit category=misc]
\tikzstyle{variable photon}=[inner sep=0mm, minimum size=5mm, shape=circle, draw=blue, fill=white, tikzit category=misc]

\tikzstyle{hadamard edge}=[-, dashed, dash pattern=on 2pt off 0.5pt, thick, draw={rgb,255: red,68; green,136; blue,255}, fill=none]
\tikzstyle{box edge}=[-, dashed, dash pattern=on 2pt off 0.5pt, thick, draw={rgb,255: red,203; green,192; blue,225}]
\tikzstyle{Y fusion}=[-, dashed, dash pattern=on 2pt off 0.5pt, thick, draw={rgb,255: red,68; green,136; blue,255}]
\tikzstyle{X fusion}=[-, dash pattern=on 2pt off 0.5pt, line width=1.6pt, shorten <=-0.17mm, shorten >=-0.17mm, draw=black, fill=none, tikzit draw={rgb,255: red,21; green,93; blue,5}]
\tikzstyle{X box}=[X phase dot, shape=rectangle, tikzit shape=rectangle, tikzit draw=blue, fill={rgb,255: red,255; green,136; blue,136}, font={\footnotesize\boldmath}, tikzit category=zx]
\tikzstyle{diredge}=[->]
\tikzstyle{highlighted edge}=[-, blue, very thick]

%% file: diagrams/x_fusion_notation.tikz
\begin{tikzpicture}
	\begin{pgfonlayer}{nodelayer}
		\node [style=vertex] (0) at (-0.75, 0) {};
		\node [style=vertex] (1) at (0, 0) {};
		\node [style=vertex] (2) at (0.75, 0) {};
		\node [style=vertex] (3) at (-0.75, -1) {};
		\node [style=vertex] (4) at (0, -1) {};
		\node [style=vertex] (5) at (0.75, -1) {};
		\node [style=none] (6) at (1.5, -0.5) {:=};
		\node [style=vertex] (7) at (2.25, 0) {};
		\node [style=vertex] (8) at (3, 0) {};
		\node [style=vertex] (9) at (3.75, 0) {};
		\node [style=vertex] (10) at (2.25, -1) {};
		\node [style=vertex] (11) at (3, -1) {};
		\node [style=vertex] (12) at (3.75, -1) {};
		\node [style=vertex] (13) at (3.5, -0.5) {};
		\node [style=none] (14) at (3.75, -0.5) {$X$};
		\node [style=none] (15) at (4.5, -0.5) {=};
		\node [style=vertex] (16) at (5.25, 0) {};
		\node [style=vertex] (18) at (6.75, 0) {};
		\node [style=vertex] (19) at (5.25, -1) {};
		\node [style=vertex] (20) at (6, -0.5) {};
		\node [style=vertex] (21) at (6.75, -1) {};
	\end{pgfonlayer}
	\begin{pgfonlayer}{edgelayer}
		\draw (5) to (4);
		\draw (4) to (3);
		\draw (0) to (1);
		\draw (1) to (2);
		\draw [style=X fusion] (1) to (4);
		\draw (12) to (11);
		\draw (11) to (10);
		\draw (7) to (8);
		\draw (8) to (9);
		\draw (11) to (13);
		\draw (13) to (8);
		\draw (21) to (20);
		\draw (20) to (19);
		\draw (16) to (20);
		\draw (20) to (18);
	\end{pgfonlayer}
\end{tikzpicture}

%% file: diagrams/y_fusion_notation.tikz
\begin{tikzpicture}
	\begin{pgfonlayer}{nodelayer}
		\node [style=vertex] (0) at (0.75, 0) {};
		\node [style=vertex] (1) at (1.5, 0) {};
		\node [style=vertex] (2) at (2.25, 0) {};
		\node [style=vertex] (3) at (0.75, -1) {};
		\node [style=vertex] (4) at (1.5, -1) {};
		\node [style=vertex] (5) at (2.25, -1) {};
		\node [style=none] (6) at (3, -0.5) {:=};
		\node [style=vertex] (7) at (3.75, 0) {};
		\node [style=vertex] (8) at (4.5, 0) {};
		\node [style=vertex] (9) at (5.25, 0) {};
		\node [style=vertex] (10) at (3.75, -1) {};
		\node [style=vertex] (11) at (4.5, -1) {};
		\node [style=vertex] (12) at (5.25, -1) {};
		\node [style=vertex] (13) at (5, -0.5) {};
		\node [style=none] (14) at (5.25, -0.5) {$Y$};
		\node [style=none] (15) at (6, -0.5) {=};
		\node [style=vertex] (16) at (6.75, 0) {};
		\node [style=vertex] (17) at (8.25, 0) {};
		\node [style=vertex] (18) at (6.75, -1) {};
		\node [style=vertex] (20) at (8.25, -1) {};
		\node [style=vertex] (21) at (7.5, -1) {};
		\node [style=vertex] (22) at (7.5, 0) {};
		\node [style=none] (23) at (4.5, 0.375) {$\frac{\pi}{2}$};
		\node [style=none] (24) at (4.5, -1.375) {$\frac{\pi}{2}$};
	\end{pgfonlayer}
	\begin{pgfonlayer}{edgelayer}
		\draw (5) to (4);
		\draw (4) to (3);
		\draw (0) to (1);
		\draw (1) to (2);
		\draw (12) to (11);
		\draw (11) to (10);
		\draw (7) to (8);
		\draw (8) to (9);
		\draw (11) to (13);
		\draw (13) to (8);
		\draw (18) to (21);
		\draw (21) to (20);
		\draw (17) to (22);
		\draw (22) to (16);
		\draw (22) to (21);
		\draw [style=Y fusion] (4) to (1);
	\end{pgfonlayer}
\end{tikzpicture}

%% file: 1-problem.tex
We define $XY$-fusion networks using the framework established in~\cite{felice_fusion_2024}.

\begin{definition}[Fusion Network~\cite{felice_fusion_2024}]\label{def_fusion_network}
An XY-fusion network $\M{F} = (G, I, O, X, Y, \lambda, \alpha)$ is specified by:
\begin{enumerate}[1.]
\item a labelled open graph $(G, I, O, \lambda, \alpha)$ where $G = (V, E)$, and
\item sets of unordered pairs $X, Y \subseteq \{\{a, b\}\ |\ a, b \in V\}$ corresponding to X fusions and Y fusions on the nodes respectively.
\end{enumerate}
\end{definition}

We define \textit{$X$-fusion networks} and \textit{$Y$-fusion networks} as $XY$-fusion networks comprised solely $X$ fusions and $Y$ fusions respectively.
A \textit{linear $XY$-fusion network} is one whose corresponding graph is a disjoint union of lines.

We say that the $XY$-fusion network $\M{F} = (G, I, O, X, Y, \lambda, \alpha)$ \textit{implements the labelled open graph} $(G^\prime, I, O, \lambda, \alpha)$ where $G^\prime$ is defined as
\[
G^\prime = \frac{(V, E + Y)}{u \sim v \tx{ if } \{u, v\} \in X}
\]
where $+$ denotes the disjoint union. Informally, $G^\prime$ is constructed by converting $Y$ fusions into edges and merging vertices that belong to $X$ fusions.

The graph below illustrates how a target graph can be implemented using an $XY$ fusion network, 
where dotted edges represent $X$ fusions and dashed blue edges represent $Y$ fusions.

\begin{figure}[H]
\centering
\input{diagrams/linear_xy_fusion_network.tikz}
\end{figure}

We refer to such a collection of trails as a \textit{trail cover} and show that every linear XY fusion network is in one-to-one correspondance with a trail cover.

\begin{definition}[Trail cover]\label{def_trail_cover}
A \textit{trail} in a graph is a sequence of vertices where adjacent vertices are adjacent in the graph and edges may not be repeated. The length of a trail $T$, denoted $|T|$, is the number of edges in the trail.
A \textit{trail cover} of a graph $G$ is a set of edge-disjoint trails that traverse each vertex of $G$ at least once.
\end{definition}

A linear $XY$ fusion network corresponds to a trail cover of $G$ where each trail represents a linear resource state. 
For every vertex traversed by multiple trails, we employ an $X$ fusion to merge the corresponding nodes. 
We use a $Y$ fusion for every edge in the graph that is absent from the trail cover.

In the special case of linear fusion networks using only $Y$ fusions, trails cannot intersect at vertices,
and consequently the trail cover constitutes a path cover of the graph.

\begin{definition}[Path cover]
A \textit{path} in a graph is a sequence of vertices such that adjacent vertices in the sequence are adjacent in the graph, and vertices are not repeated.
A \textit{path cover} of $G$ is a set of vertex-disjoint paths where every vertex in $G$ belongs to exactly one path in the cover.
\end{definition}

Linear $X$ fusion networks lack $Y$ fusions to add edges, so every edge in the target graph state must be implemented by a resource state. This constraint means the corresponding trail cover is a trail decomposition of the graph.

\begin{definition}[Trail Decomposition]\label{def_trail_decomposition}
A \textit{trail decomposition} of a graph is a set of edge-disjoint trails that together traverse every edge of the graph.
\end{definition}

Given a graph $G$ and a trail cover $\M{C}$, we denote the number of intersections between trails as $X(G, \M{C})$ (where $k$ trails traversing the same vertex contribute $k-1$ intersections) and the number of edges not traversed by any trail as $Y(G,\M{C})$.
Equivalently, $X(G, \M{C})$ and $Y(G, \M{C})$ represent the number of $X$ and $Y$ fusions required to implement $G$ using trail cover $\M{C}$, respectively.
We can now state the main theorem of this section.

\begin{theorem}\label{thm_num_fusions}
Given a graph $G = (V, E)$ and trail cover $\M{C}$ of $G$:
\[
X(G, \M{C}) + Y(G, \M{C}) = |E| - |V| + |\M{C}|.
\]
\end{theorem}
\begin{proof}
Let $T \in \M{C}$ be a trail and let $E(T)$ denote the number of edges in $T$ and $V(T)$ denote the number of vertices in $T$ (counting repeated vertices multiple times). Then $V(T) = E(T) + 1$, and summing over all trails yields
\begin{equation}\label{eq_sum_over_trails}
\sum_{T \in \M{C}} V(T) = \sum_{T \in \M{C}} (E(T) + 1) = \sum_{T \in \M{C}} E(T) + |\M{C}|.
\end{equation}
By the definitions of $Y(G, \M{C})$ and $X(G, \M{C})$, we have $\sum_{T \in \M{C}} E(T) = |E| - Y(G, \M{C})$ and $\sum_{T \in \M{C}} V(T) = |V| + X(G, \M{C})$.

Substituting into~\eqref{eq_sum_over_trails} gives
\[
X(G, \M{C}) + Y(G, \M{C}) = |E| - |V| + |\M{C}|.
\]
\end{proof}

This result naturally extends to path covers and trail decompositions, which are the special cases where $X(G, \M{C}) = 0$ and $Y(G, \M{C}) = 0$ respectively.
For a path cover $\M{C}$ of $G$, we have
\[
Y(G, \M{C}) = |E| - |V| + |\M{C}|.
\]
For a trail decomposition $\M{C}$ of $G$, we have
\[
X(G, \M{C}) = |E| - |V| + |\M{C}|.
\]

A direct consequence of Theorem~\ref{thm_num_fusions} is that minimizing the total number of $X$ and $Y$ fusions in a trail cover is equivalent to minimizing the number of trails.
This observation leads us to the following optimization problems.

\begin{definition}
\problem{MinTrailCover}

\noindent\textbf{\textit{Input:}} A graph $G$.\\
\textbf{\textit{Output:}} A trail cover of $G$ with the minimum number of trails.
\end{definition}

\begin{definition}
\problem{MinPathCover}

\noindent\textbf{\textit{Input:}} A graph $G$.\\
\textbf{\textit{Output:}} A path cover of $G$ with the minimum number of paths.
\end{definition}

\begin{definition}\label{def_mtd}
\problem{MinTrailDecomposition}

\noindent\textbf{\textit{Input:}} A graph $G$.\\
\textbf{\textit{Output:}} A trail decomposition of $G$ with the minimum number of trails.
\end{definition}

In practice, resource states have constant or bounded length due to physical constraints. 
Photonic approaches using SPDC sources generate constant-size resource states~\cite{cite_spdc_ghz}.
Matter-based approaches emit linear resources of bounded length, which depends on the coherence time of the atom~\cite{cite_deterministic_solid_state_emitter, cite_matter_based_linear_cluster_state_generation_2}.
These practical limitations motivate bounded variants of the above problems.
A trail containing at most $L$ edges is called an \textit{$L$-trail}, and an \textit{$L$-trail cover} is a trail cover consisting entirely of $L$-trails.
$L$-paths, $L$-path covers, and $L$-trail decompositions are defined analogously. This leads to the following optimization problems.

\begin{definition}
\problem{MinBoundedTrailCover}

\noindent\textbf{\textit{Input:}} A graph $G$ and integer $L$.\\
\textbf{\textit{Output:}} An $L$-trail cover of $G$ with the minimum number of trails.
\end{definition}

\begin{definition}
\problem{MinBoundedPathCover}

\noindent\textbf{\textit{Input:}} A graph $G$ and integer $L$.\\
\textbf{\textit{Output:}} An $L$-path cover of $G$ with the minimum number of paths.
\end{definition}

\begin{definition}
\problem{MinBoundedTrailDecomposition}

\noindent\textbf{\textit{Input:}} A graph $G$ and integer $L$.\\
\textbf{\textit{Output:}} An $L$-trail decomposition of $G$ with the minimum number of trails.
\end{definition}

We note that constraining the linear resource states by their photon count would more accurately capture physical hardware limitations.
In most photonic architectures, each fusion between two nodes requires one photon from each node, and each node also requires a photon for single-qubit measurement.
For any resource state implementing a trail $T$ in a trail cover, we can assign weight $w(T) = L + F$ where $L$ is the number of nodes in the trail and $F$ is the number of fusions involving these nodes. 
The corresponding problem, denoted \problem{MinPhotonBoundedTrailCover}, seeks a trail cover of $G$ with minimum trail count such that every trail's weight is bounded by a constant.
For atom-based implementations of linear resource states (such as quantum dots), this constant would be proportional to the atom's coherence time~\cite{cite_deterministic_catapillar_states}.
In Section~\ref{sec:complexity}, we show that the bounded minimum trail decomposition and trail covering problems have are in the same complexity class as the photon-bounded variants. 
Moreover, the approximation algorithms for bounded trail covers and decompositions developed in Section~\ref{sec:approximation} 
are easily generalised to the photon-bounded setting and were used to produce the results in Section~\ref{sec:benchmarks}.

%% file: diagrams/linear_xy_fusion_network.tikz
\begin{tikzpicture}
	\begin{pgfonlayer}{nodelayer}
		\node [style=vertex] (0) at (-1, 2) {};
		\node [style=vertex] (1) at (-1, 1) {};
		\node [style=vertex] (2) at (0, 1) {};
		\node [style=vertex] (3) at (0, 2) {};
		\node [style=vertex] (4) at (0, 3) {};
		\node [style=vertex] (5) at (1, 3) {};
		\node [style=vertex] (6) at (1, 2) {};
		\node [style=vertex] (7) at (3, 2.25) {};
		\node [style=vertex] (8) at (2.75, 0.75) {};
		\node [style=vertex] (9) at (4.5, 0.75) {};
		\node [style=vertex] (10) at (4, 2.25) {};
		\node [style=vertex] (11) at (4, 3.25) {};
		\node [style=vertex] (12) at (5.5, 3.25) {};
		\node [style=vertex] (13) at (5.5, 1.75) {};
		\node [style=vertex] (14) at (2.75, 1.75) {};
		\node [style=vertex] (15) at (3.75, 0.75) {};
		\node [style=vertex] (16) at (4.5, 1.75) {};
		\node [style=none] (17) at (2, 2) {$\to$};
	\end{pgfonlayer}
	\begin{pgfonlayer}{edgelayer}
		\draw (2) to (6);
		\draw (6) to (5);
		\draw (5) to (4);
		\draw (4) to (3);
		\draw (3) to (6);
		\draw (2) to (3);
		\draw (3) to (0);
		\draw (0) to (1);
		\draw (1) to (2);
		\draw (13) to (12);
		\draw (11) to (10);
		\draw (10) to (7);
		\draw (15) to (8);
		\draw (8) to (14);
		\draw (16) to (9);
		\draw (16) to (13);
		\draw [style=Y fusion] (9) to (13);
		\draw [style=Y fusion] (11) to (12);
		\draw [style=X fusion] (16) to (10);
		\draw [style=X fusion] (15) to (9);
		\draw [style=X fusion] (14) to (7);
	\end{pgfonlayer}
\end{tikzpicture}

%% file: 2-complexity.tex
\subsection{The minimum path cover problem is NP-hard}

The minimum path cover problem is NP-hard since the existence of a path cover containing a single path indicates that the graph contains a Hamiltonian path, and deciding whether a graph has a Hamiltonian path is NP-complete~\cite{cite_karp_21}. Since verifying whether a given path cover is minimal requires solving the optimization problem itself, \problem{MinPathCover} is NP-hard. Special graph classes for which the Hamiltonian path problem admits polynomial-time solutions are cataloged in~\cite{cite_catalog_hamiltonian} and may admit efficient algorithms for finding minimum path covers.

Moran et al.~\cite{cite_approx_mpc} developed a $\frac{2}{3}$-approximation algorithm for finding path covers on weighted graphs that maximize total weight. This is equivalent to the minimum path cover problem when all edges have unit weight.
This result also implies that \problem{MinPhotonBoundedPathCover} is NP-hard.
Kobayashi et al.~\cite{cite_mpc} generalized this to consider path covers where each path carries a weight based on its length. 
Setting all path weights to one regardless of length yields the bounded minimum path cover problem. 
The authors showed that the bounded minimum path cover problem is solvable in polynomial time for graphs with bounded tree width.
This gives an algorithm for solving the minimum $L$-path cover problem in time $O(2^{2W} W^{2W + 2} (L + 2)^{2W + 2} |V|)$ when $G$ has tree-width smaller than $W$.
We present no new results on path cover problems and instead provide graph rewrites in Section~\ref{sec:graphrewrites} to heuristically reduce the number of $Y$ fusions and the size of the minimum path covers.

\subsection{The minimum trail decomposition problem is in P}

Fortunately, an efficient algorithm exists for finding minimum trail decompositions by reducing the problem to finding Eulerian circuits.

\begin{lemma}[Euler~\cite{cite_eulers_theorem}]\label{lem_euler}
A connected graph has an Eulerian circuit if and only if every vertex has even degree.
\end{lemma}

\begin{definition}
Let $G = (V, E)$ be a graph. Then $\Odd(G) \subseteq V$ is the set of all vertices of $G$ that have odd degree. We say such vertices are \textit{odd} and all other vertices are \textit{even}.
\end{definition}

\begin{theorem}[Theorem 2.3~\cite{cite_orig_proof_of_mtd}]\label{thm_min_trail_decomp}
Let $G$ be a connected graph. Then there exists a minimum trail decomposition of $G$ that has $\frac{1}{2}|\Odd(G)|$ trails if $|\Odd(G)| > 0$ and a single trail otherwise.
\end{theorem}

\begin{proof}
This bound is minimal since each odd vertex must serve as an endpoint for some trail.

Every graph contains an even number of odd vertices. Therefore we can construct a trail decomposition achieving this bound by connecting odd vertices randomly with an edge. The resulting graph has no odd vertices, and so we can find an Eulerian circuit by Lemma~\ref{lem_euler}. Removing the introduced edges from the circuit creates $\frac{1}{2}|\Odd(G)|$ trails if $|\Odd(G)| > 0$ and one trail otherwise. These trails cover every edge of the original graph and therefore constitute a trail decomposition.
\end{proof}

We conclude that the minimum trail decomposition problem can be solved in polynomial time since efficient algorithms exist for finding Eulerian circuits in time $O(|E|)$, such as Hierholzer's algorithm~\cite{cite_hierholzer_alg}.

\begin{theorem}\label{thm_mtd_in_p}
\problem{MinTrailDecomposition} is in $P$.
\end{theorem}

Minimum trail decompositions possess a particularly useful structure that enables efficient testing of whether a trail belongs to some minimum trail decomposition. We use this characterization extensively to prove subsequent results and when developing heuristic algorithms for finding trail covers in Section~\ref{sec:approximation}.

Theorem~\ref{thm_min_trail_decomp} generalizes to disconnected graphs as follows.

\begin{lemma}\label{lem_mtd_num_cc}
Let $G$ be a possibly disconnected graph. The minimum trail decomposition of $G$ contains at least $\frac{1}{2}|\Odd(G)|$ trails, with equality if and only if every connected component of $G$ has non-zero odd vertices.
\end{lemma}

For a trail $T$ in graph $G$, we write $G \backslash T$ to denote the subgraph of $G$ with the edges of $T$ removed.
We now state the necessary and sufficient conditions for a trail to belong to a minimum trail decomposition.

\begin{proposition}\label{prop_mtd_conditions}
Let $G$ be a connected graph with non-zero odd vertices. A trail $T$ in $G$ belongs to some minimum trail decomposition of $G$ if and only if $T$ begins and ends at distinct odd vertices and every connected component of $G \backslash T$ has non-zero odd vertices.
\end{proposition}

\begin{proof}
Suppose $T$ belongs to a minimum trail decomposition $\M{T}$ of $G$. Then $\M{T} \backslash \{T\}$ must be a minimum trail decomposition for $G \backslash T$. Therefore by Theorem~\ref{thm_min_trail_decomp}:
\[
|\Odd(G \backslash T)| = 2|\M{T} \backslash \{T\}| = 2(|\M{T}| - 1) = |\Odd(G)| - 2.
\]
Removing $T$ from $G$ can only decrease the number of odd vertices by two if $T$ connects distinct odd vertices. Lemma~\ref{lem_mtd_num_cc} ensures that every connected component of $G \backslash T$ has non-zero odd vertices.

Conversely, assume trail $T$ ends at distinct odd vertices. Then $|\Odd(G \backslash T)| = |\Odd(G)| - 2$ by Lemma~\ref{lem_mtd_num_cc}. If every connected component of $G \backslash T$ contains at least one odd vertex, then Lemma~\ref{lem_mtd_num_cc} guarantees the existence of a minimum trail decomposition $\M{T}^\prime$ of $G \backslash T$ of size $\frac{1}{2}|\Odd(G \backslash T)| = \frac{1}{2}|\Odd(G)| - 1$. Then $\M{T} = \M{T}^\prime \cup \{T\}$ is a trail decomposition of $G$ of size $\frac{1}{2}|\Odd(G)|$, which by Theorem~\ref{thm_min_trail_decomp} is minimal.
Therefore $T$ belongs to a minimum trail decomposition.
\end{proof}

\subsection{The bounded minimum trail decomposition problem is NP-hard}

We now consider the \problem{MinBoundedTrailDecomposition} problem, which seeks a minimum trail decomposition with trails of length at most $L \ge 1$.

One might expect there to always exist an $L$-trail decomposition of size $|E|/L$ or $\frac{1}{2}|\Odd(G)|$, but this is not always the case, as shown in the example below.
\begin{figure}[H]
\centering
\input{diagrams/awkward_to_decompose_graph.tikz}
\caption{An example of a graph that whose minimum $L$-trail decomposition contains more than $|E|/L$ or $\frac{1}{2}|\Odd(G)|$ trails. For $L = 3$, the minimum 3-trail decomposition requires 3 trails but $|E|/3 = 2$ and $\frac{1}{2}|\Odd(G)| = 2$.}
\end{figure}
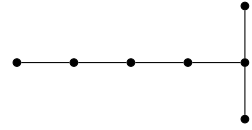

The minimum $1$-trail decomposition is simply the edge set $E$. There also exists an efficient algorithm for computing a minimum $2$-trail decomposition by converting it to a matching problem.

\begin{proposition}\label{prop_2trail}
Let $G = (V, E)$ be a graph. Then the minimum 2-trail decomposition of $G$ contains $\ceil{\frac{1}{2}|E|}$ trails and can be found in polynomial time.
\end{proposition}

\begin{proof}
A trail decomposition of size $\ceil{\frac{1}{2}|E|}$ consists entirely of 2-trails plus a single 1-trail when $|E|$ is odd, and is therefore a lower bound.
Observe that 2-trails in $G$ correspond to matchings in the line graph $L(G)$ of $G$.
Therefore a minimum 2-trail decomposition maximizes the number of 2-trails, which corresponds to a maximum matching on $L(G)$.
Since line graphs are connected and claw-free, $L(G)$ has a perfect matching when the line graph has an even number of vertices~\cite{cite_claw_free_matching}, which occurs when $G$ has an even number of edges.

When $G$ has an even number of edges, the line graph admits $\frac{1}{2}|E|$ matches, yielding a 2-trail decomposition with $\frac{1}{2}|E|$ trails. 

When $G$ has an odd number of edges, we can remove a non-bridge edge to obtain a graph with an even number of edges and find a perfect matching of size $\frac{1}{2}(|E| - 1)$ in the line graph.
Mapping each match to a 2-trail in the original graph and implementing the removed edge with a 1-trail gives a 2-trail decomposition of size $\frac{1}{2}(|E| + 1) = \ceil{\frac{1}{2}|E|}$.
Thus a minimum 2-trail decomposition of $G$ has size $\ceil{\frac{1}{2}|E|}$.

Maximal cardinality matchings can be found in polynomial time~\cite{cite_poly_matching} and therefore minimum 2-trail decompositions can also be found in polynomial time.
\end{proof}

The situation becomes more complex when $L \ge 3$ and is in fact NP-hard. To show this, we first prove that the corresponding decision problem is NP-complete.

\begin{theorem}\label{thm_trail_decomp_np_complete}
Given a graph $G = (V, E)$, determining whether there exists an $L$-trail decomposition of $G$ of size $K \ge 1$ is NP-complete.
\end{theorem}

\begin{proof}
Given a solution, we can verify whether it is a valid $L$-trail decomposition with $K$ trails in polynomial time, so the problem is in NP. The rest of the proof follows by constructing a polynomial reduction from the bin packing problem which is known to be NP-complete.

The bin packing problem can be stated as follows: given a set of $n$ items with integer weights $(w_i)_{i=1}^n$ where $w_i > 1$ for all $i$, and positive integers $C$ and $K$. Determine whether there exists a partition of the items into $K$ disjoint sets $p_1, \ldots, p_K$ such that the total weight of all items in each partition is at most $C$, that is, $\sum_{j \in p_k} j \le C$ for all $1 \le k \le K$.

We begin by constructing a graph $G$ with $2K$ vertices $\{u_i\}_{i=1}^{2K}$, all connected to another vertex $v$. For each item $j$, we construct a cycle of $w_j$ edges starting and ending at $v$. Since all weights are greater than one, no self-loops exist and the graph is simple.

\begin{figure}[H]
\centering
\input{diagrams/loop_construction.tikz}
\end{figure}

We claim that solutions to the original bin packing problem are in one-to-one correspondence with $(C+2)$-trail decompositions of size $K$ for $G$.

Suppose we have a $(C + 2)$-trail decomposition $\M{T}$ of $G$ of size $K$. Observe that each trail in $\M{T}$ must start and end at one of the vertices in $\{u_i\}_{i=1}^{2K}$. Therefore every trail either fully traverses a particular cycle or doesn't traverse any edge in the cycle. Since the trail has length at most $C+2$, subtracting the first and last edge of the trail between $v$ and its endpoints in $\{u_i\}_{i=1}^{2K}$, the sum of the edges of the cycles it traverses must not exceed $C$.

Each trail therefore corresponds to a partition of the items, namely the items associated with the cycles it traverses that solves the original bin-packing problem. Conversely, it is straightforward to see that any partition of the items can be converted into a $(C+2)$-trail decomposition of the graph by mapping each partition to a trail that begins and ends at one of the vertices in $\{u_i\}_{i=1}^{2K}$, and traverses every cycle associated with the items in the partition.

Therefore since the bin packing problem is NP-complete, the $L$-trail decomposition decision problem is also NP-complete.
\end{proof}

Given that the decision problem is NP-complete, the corresponding minimization problem is therefore NP-hard.

\begin{theorem}\label{thm_min_bounded_td_nphard}
\problem{MinBoundedTrailDecomposition} is NP-hard.
\end{theorem}

Approximating the bin packing problem with ratio smaller than $\frac{3}{2}$ is NP-hard~\cite{cite_bin_pack_approx}. This constraint applies equally to the minimum $L$-decomposition problem, though we show in Section~\ref{sec:approximation} that the accuracy of our approximation algorithm depends on the number of odd vertices in the graph and provide a heuristic algorithm returning an $L$-trail decomposition containing on average $\frac{1}{4}|\Odd(G)|$ more trails than the minimum bounded trail decomposition.

\problem{MinPhotonBoundedTrailDecomposition} is more complex since many possible $X$ fusion arrangements exist for merging nodes into single vertices, affecting photon counts in each resource state. The trail decomposition alone doesn't contain sufficient information to determine solution validity. We can address this issue by scaling the graph $G$ from Theorem~\ref{thm_trail_decomp_np_complete} by a factor large enough to make any $X$ fusion arrangement irrelevant, then applying the same proof.

\begin{lemma}\label{lem_num_fusions}
The constructed graph $G$ in the proof of Theorem~\ref{thm_trail_decomp_np_complete} contains $n + K - 1$ fusions.
\end{lemma}

\begin{proof}
Loop $i$ has $w_j$ edges, and so including the edges from the $2K$ nodes in $U$, there are $2K + \sum_{j=1}^n w_j$ edges in the graph. There are $\sum_{j=1}^n (w_j - 1)$ vertices in cycles excluding the common point $v$. If we count $v$ and the $2K$ points in $U$, we have $2K + 1 - n + \sum_{j=1}^n w_j$ vertices. Therefore any minimum trail decomposition has $K$ trails, so we have
\begin{align*}
F &= |E| - |V| + K \\
  &= 2K + \sum_{j=1}^n w_j - (2K + 1 - n + \sum_{j=1}^n w_j) + K \\
  &= n + K - 1
\end{align*}
\end{proof}

\begin{lemma}\label{lem_photons_in_trail}
For any subset $\{s_i\}_{i=1}^p$ of the items, the trail in the graph corresponding to the subset has $P$ photons where
\[
2K + \sum_{j=1}^n w_j \le P \le 4K + 2n - 2 - p + \sum_{j=1}^n w_j.
\]
\end{lemma}

\begin{proof}
We have $2K$ measurement photons in $U$, optionally one measurement photon in $v$, $\sum_{j=1}^n (w_j-1) = (\sum_{j=1}^n w_j) - p$ measurement photons in the cycles, and between $p$ and $2F = 2(n + K - 1)$ fusion photons by Lemma~\ref{lem_num_fusions}.

We can write this as
\[
P = 2K + F_p - p + \sum_{j=1}^n w_j
\]
where $F_p$ is the number of fusion photons in the cycle and $p \le F_p \le 2(n + K - 1)$. Substituting the lower and upper bounds for $F_p$ yields the result.
\end{proof}

\begin{theorem}
\problem{MinPhotonBoundedTrailDecomposition} is NP-hard.
\end{theorem}

\begin{proof}
Our proof follows by showing that there exists a constant positive integer $S$ such that by scaling the cycles of the graph in the previous proof to have $S$ times more edges we will obtain a graph where the solutions to \problem{MinPhotonBoundedTrailDecomposition} are in one-to-one correspondance with solutions to the original bin packing problem.

First note that by Lemma~\ref{lem_photons_in_trail} we have that any subset of the items $\{i_j\}_{j=1}^n$ where $\sum_{j=1}^n w_{i_j} \le C$ for some constant $C$, then the corresponding trail in the graph has at most $(SC - p + 4K + 2n - 2)$ photons and so any solution to the bin packing problem is also a solution to the $(SC - p + 4K + 2n - 2)$-photon bounded trail decomposition problem.

Now suppose the subset of items sums to $C + 1$ and is hence not in a valid solution to the bin packing problem. Then the minimum number of photons the corresponding trail can have is $2K + S(C + 1)$ by Lemma~\ref{lem_photons_in_trail}. If we were to make this trail not part of a valid solution to the same trail cover problem, then
\[
2K + S(C + 1) > SC - p + 4K + 2n - 2
\]
Rearranging gives
\[
S > 2K + 2n - 2 - p
\]
Thus for sufficiently large $S$, solutions to the bin packing problem are in one-to-one correspondance with the solutions to the $(SC - p + 4K + 2n - 2)$-photon bounded trail decomposition problem. Thus \problem{MinPhotonBoundedTrailDecomposition} is NP-hard.
\end{proof}

\subsection{The minimum trail cover problem is NP-hard}

We might expect the minimum trail cover problem to be at least as hard as the minimum path cover problem. We confirm this intuition by showing that solutions to \problem{MinTrailCover} can produce solutions to \problem{MinPathCover} on cubic graphs and is therefore NP-hard.

\begin{theorem}\label{thm_mtc_nphard}
\problem{MinTrailCover} is NP-hard.
\end{theorem}

\begin{proof}
Let $G$ be a cubic graph and suppose $\M{C}$ is a minimum trail cover for $G$. Then each vertex is traversed by at most two trails since if it was traversed by three, then all three must eand at the vertex and so we may join two trails together to obtain a smaller trail cover. For any vertex traversed by two trails in $\M{C}$, one of the trails must end at the vertex since the degree of the vertex is three. We can then removing the last edge of this trail producing a trail cover of the same size. By performing these retractions wherever possible, we obtain a trail cover where each vertex is traversed by exactly one trail. Thus each trail is a path and the trail cover is now a minimum path cover of $G$.

Since finding a Hamiltonian path in a cubic graph is NP-hard~\cite{cite_cubic_hamiltonian_nphard}, finding a minimum path cover is also NP-hard and therefore \problem{MinTrailCover} is NP-hard.
\end{proof}

It follows naturally that \problem{MinBoundedTrailCover} is also NP-hard.

This also implies that the corresponding graph problem for minimizing fusion networks with photon-bounded linear resource states is NP-hard.

\begin{theorem}
    \problem{MinPhotonBoundedTrailCover} is NP-hard.
\end{theorem}
\begin{proof}
    Suppose we have a trail $T$ in a cubic graph where each trail corresponds to a linear resource state of length $L$. Then by the proof of Theorem~\ref{thm_mtc_nphard}, we see that we take $T$ to be a path. Then each node in the resource state has one photon for a measurement (or for output) and one photon for every fusion that occurs at the node. Since $T$ is in a cubic graph, each intermediate node has one fusion and each endpoint has two fusions. Thus the resource state consists of at most $2(L - 1) + 2*3 = 2L + 4$ photons. Therefore a solution to the bounded photon fusion network problem of size $2L + 4$ is a solution to the minimum $L$-trail cover on cubic graphs which we know to be NP-hard from the proof of Theorem~\ref{thm_mtc_nphard}. Observe that if we had $2L + 5$ photons, the size of our resource states would not change since any additional intermediate node required two photons. Therefore we can conclude this problem is also NP-hard.
\end{proof}

%% file: diagrams/awkward_to_decompose_graph.tikz
\begin{tikzpicture}
	\begin{pgfonlayer}{nodelayer}
		\node [style=vertex] (0) at (-1, 0) {};
		\node [style=vertex] (1) at (-0.25, 0) {};
		\node [style=vertex] (2) at (0.5, 0) {};
		\node [style=vertex] (3) at (1.25, 0) {};
		\node [style=vertex] (4) at (2, 0) {};
		\node [style=vertex] (5) at (2, 0.75) {};
		\node [style=vertex] (6) at (2, -0.75) {};
	\end{pgfonlayer}
	\begin{pgfonlayer}{edgelayer}
		\draw (6) to (4);
		\draw (4) to (5);
		\draw (4) to (3);
		\draw (3) to (2);
		\draw (2) to (1);
		\draw (1) to (0);
	\end{pgfonlayer}
\end{tikzpicture}

%% file: diagrams/loop_construction.tikz
\begin{tikzpicture}
	\begin{pgfonlayer}{nodelayer}
		\node [style=vertex] (0) at (0, 0) {};
		\node [style=vertex] (1) at (-0.75, -1) {};
		\node [style=vertex] (2) at (0.75, -1) {};
		\node [style=none] (3) at (0.025, -0.75) {$\cdots$};
		\node [style=vertex] (5) at (-1.5, 0.25) {};
		\node [style=vertex] (6) at (-0.75, 1.25) {};
		\node [style=none] (7) at (-1.65, 0.575) {};
		\node [style=none] (8) at (-1.25, 1.25) {};
		\node [style=none] (9) at (-1.5, 1) {\rotatebox{45}{$\cdots$}};
		\node [style=none] (10) at (-2, 1.25) {$w_1$};
		\node [style=none] (11) at (0.025, 1) {$\cdots$};
		\node [style=vertex] (12) at (0.75, 1.25) {};
		\node [style=vertex] (13) at (1.5, 0.25) {};
		\node [style=none] (14) at (1.25, 1.25) {};
		\node [style=none] (15) at (1.65, 0.575) {};
		\node [style=none] (16) at (1.5, 1) {\rotatebox{135}{$\cdots$}};
		\node [style=none] (17) at (2, 1.25) {$w_n$};
		\node [style=none] (18) at (0, -0.325) {v};
		\node [style=none] (19) at (0.75, -1.5) {$u_{2K}$};
		\node [style=none] (20) at (-0.75, -1.5) {$u_1$};
	\end{pgfonlayer}
	\begin{pgfonlayer}{edgelayer}
		\draw (0) to (1);
		\draw (0) to (2);
		\draw [bend left=15] (6) to (0);
		\draw [bend right=15] (5) to (0);
		\draw [bend right=15] (7.center) to (5);
		\draw [bend left=15] (8.center) to (6);
		\draw [bend left=15] (13) to (0);
		\draw [bend right=15] (12) to (0);
		\draw [bend left=15] (15.center) to (13);
		\draw [bend left=345] (14.center) to (12);
	\end{pgfonlayer}
\end{tikzpicture}

%% file: 3-approximation.tex
\subsection{Approximating minimum $L$-trail decompositions}

Since the minimum $L$-trail decomposition problem is NP-hard (Theorem~\ref{thm_min_bounded_td_nphard}), we will now focus on developing efficient approximation algorithms. A natural approach is to find a minimum trail decomposition of unbounded length and subdivide the trails into $L$-trails. We prove tight bounds on this algorithm's accuracy and show that its accuracy depends linearly on the number of odd vertices.

The following number theory result is proved in the appendix.

\begin{restatable}{lemma}{lemNumberTheoryBounds}\label{lem_number_theory_bounds}
Let $(t_i)_{i=1}^N$ and $L$ be positive integers. Then the following inequality holds and is tight.
\[
\sum_{i = 1}^N \ceil{\frac{t_i}{L}} - \ceil{\sum_{i = 1}^N \frac{t_i}{L}} \le N - \ceil{\frac{N}{L}}.
\]
\end{restatable}

We now present our approximation algorithm and establish tight bounds on its accuracy.

\begin{restatable}{proposition}{propLTrailApprox}\label{prop_ltrail_approx}
We can find an $L$-trail decomposition of a graph $G$ in polynomial time that contains at most $\floor{\frac{1}{2}|\Odd(G)|(1 - \frac{1}{L})}$ more trails than the minimum.
\end{restatable}

\begin{restatable}{proposition}{propPhotonTrailApprox}\label{prop_photon_trail_approx}
We can find an fusion network that implements an open graph with graph $G$ in polynomial time that contains at most $\frac{1}{2}|\Odd(G)|(1 - \frac{3}{L-2}) + 1$ more resource states than the minimum.
\end{restatable}

This bound is tight, as demonstrated in the example below illustrating our subdivision algorithm.

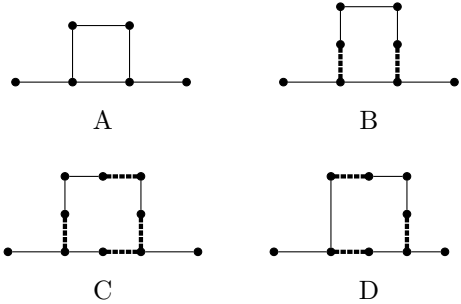
\begin{figure}[h]
\centering
\input{diagrams/example_of_reaching_bound.tikz}
\caption{Example reaching the bound in Proposition~\ref{prop_ltrail_approx}. To find a minimum $2$-trail decomposition for graph (A), we first find a minimum trail decomposition (B) and subdivide it into four $2$-trails (C). However, the minimum 2-trail decomposition has size 3 (D). This achieves the maximum error bound from Proposition~\ref{prop_ltrail_approx}: $\frac{1}{2}|\Odd(G)|(1 - \frac{1}{L}) = \frac{1}{2} \times 4 (1 - \frac{1}{2}) = 1$.}
\end{figure}

This result shows that reducing the number of odd vertices improves accuracy of the approximation. Indeed, when fewer than 3 odd vertices exist, we obtain a minimum trail decomposition.
On average, this algorithm performs twice as well as the worst case.

\begin{restatable}{proposition}{propAverageLTrailApprox}\label{prop_average_lsec:problemtrail_approx}
The result of Proposition~\ref{prop_ltrail_approx} on average contains at most $\frac{1}{4}|\Odd(G)|$ more trails than the minimum.
\end{restatable}

This subdivision algorithm easily adapts to linear resource states with bounded numbers of photons. Instead of subdividing trails based on edge count, we subdivide based on the number of photons. Each node in the resource state requires one photon per participating fusion plus one photon for measurement (or output). 

We now establish the connection between bounded minimum trail decompositions and their unbounded counterparts.

\begin{restatable}{proposition}{propConvertToMTD}\label{prop_convert_to_mtd}
Any trail decomposition can be transformed into a minimum trail decomposition by applying two rules:
\begin{itemize}
\item if two distinct trails end at the same vertex, join them together, and
\item if a closed trail traverses the same vertex as another trail, join them together
\end{itemize}
\end{restatable}

These two rules are illustrated below, where dashed lines represent nodes that may be traversed by multiple trails.
\begin{figure}[H]
\centering
\input{diagrams/joining_closed_trails.tikz}
\end{figure}

\begin{remark}
The reduction in Proposition~\ref{prop_convert_to_mtd} can prove Theorem~\ref{thm_min_trail_decomp} without invoking Euler's Theorem.
\end{remark}

We now present a key result enabling the development of better heuristics for finding minimum $L$-trail decompositions.

\begin{restatable}{proposition}{propLTrailsAreSubdivisions}\label{prop_ltrail_are_subdivision}
For any graph $G$, there exists a minimum $L$-trail decomposition of $G$ that is a subdivision of some minimum trail decomposition of $G$.
\end{restatable}

To formulate this into a heuristic, suppose we have a minimum trail decomposition $\M{T} = \{T_1, \ldots, T_K\}$ of graph $G$ where $K = \frac{1}{2}|\Odd(G)|$ if $|\Odd(G)| > 0$ and $K = 1$ otherwise. Subdividing into $L$-trails produces

\[
\sum_{i = 1}^K \ceil{\frac{|T_i|}{L}}
\]

trails. This sum is minimised when the number of trails whose length is a multiple of $L$ is maximised.

We can summarize this finding by saying that \problem{MinBoundedTrailDecomposition} is equivalent to finding an unbounded minimum trail decomposition that maximizes the number of trails whose length is a multiple of $L$. Therefore, generating all minimum trail decompositions is NP-hard, since otherwise we could use it to solve \problem{MinBoundedTrailDecomposition}.

We can encode these results into a heuristic algorithm for approximating minimum bounded trail decompositions.

\begin{algorithm}\label{alg_ltrail_decomposition}
\ \\\textit{\textbf{Input:}} A graph $G$.\\
\textit{\textbf{Output:}} An $L$-trail decomposition of $G$.
\begin{enumerate}
\item If $|\Odd(G)| \le 2$, find an Eulerian path, subdivide into $L$-trails and return.
\item Otherwise, search for a trail whose length is a multiple of $L$ and satisfies the conditions for belonging to a minimum trail decomposition in Proposition~\ref{prop_mtd_conditions}.
\item Repeat Steps 1 and 2 until no such trails can be found.
\item Remove these trails from the graph and find a minimum unbounded trail decomposition on the remaining graph.
\item Subdivide the overall trail decomposition into $L$-trails and return.
\end{enumerate}
\end{algorithm}

Algorithm~\ref{alg_ltrail_decomposition} terminates with a solution in polynomial time if the search algorithm terminates in polynomial time. Since the search algorithm may fail to find a suitable trail despite one existing, Algorithm~\ref{alg_ltrail_decomposition} has the same approximation ratio as the original subdivision algorithm in Proposition~\ref{prop_ltrail_approx}.

The search algorithm may be implemented using a breadth first search that terminates after a predetermined time if no suitable trail is found. Although the search space consists of all possible paths in a graph and is therefore exponential in size, the search algorithm can exploit the special properties of trails in minimum trail decompositions specified in Proposition~\ref{prop_mtd_conditions} to accelerate the search.

\subsection{Maximal trail covers}

To find minimum trail covers, we introduce a special category of trail covers called \textit{maximal trail covers}, which possess a richer structure that we will use to develop more efficient heuristics.

We denote the subgraph of $G$ covered by the trail $T$ as $G(T)$, and extend this notation to denote the subgraph of $G$ covered by the trail cover $\M{C}$ as $G(\M{C})$.

\begin{definition}
A trail cover $\M{C}$ of graph $G$ is \textit{maximal} if $\M{C}$ is a minimum trail decomposition of $G(\M{C})$ and for any trail cover $\M{C}^\prime$ of $G$ where $G(\M{C}) \subsetneq G(\M{C}^\prime)$, we have $|\M{C}| < |\M{C}^\prime|$.
\end{definition}

Informally, this means any trail cover that covers more of the graph than a maximal trail cover must have strictly more trails. Hence the trails in $\M{C}$ cannot be simply extended and are \textit{maximal} in this sense. 

We can use this concept to characterize the structure of edges in the graph that are not covered by trails in maximal trail covers. We first introduce two lemmas whose proofs appear in the Appendix.

\begin{restatable}{lemma}{lemOddVertices}\label{lem_odd_verts}
Let $\M{C}$ be a maximal trail cover of a connected graph $G$. Then vertices that are odd in $G(\M{C})$ are also odd in $G$ and have the same degree.
\end{restatable}

\begin{restatable}{lemma}{lemMaximalTC}\label{lem_maximal_tc}
Let $\M{C}$ be a maximal trail cover of a connected graph $G$. Then every connected component of $G(\M{C})$ has non-zero odd vertices if $G$ has non-zero odd vertices.
\end{restatable}

\begin{theorem}\label{thm_mtc_subset_mtd}
A trail cover $\M{C}$ of a graph $G$ is maximal if and only if it is a subset of a minimum trail decomposition of $G$, every the degree of every odd vertex in $G(\M{C})$ is the same as in $G$, and $G \backslash G(C)$ is acyclic.
\end{theorem}

\begin{proof}
If $G$ has no odd vertices, maximal trail covers are simply minimum trail decompositions consisting of a single Eulerian tour of $G$ and so the theorem holds. For the remainder of the proof, we consider the case where $G$ has non-zero odd vertices. We also assume $G$ is connected since a maximal trail cover on a disconnected graph is maximal on all connected components.

Let $\M{C}$ be a maximal trail cover of $G$ and let $T$ be a trail in $\M{C}$. 
From Lemma~\ref{lem_maximal_tc}, the connected component of $G(\M{C})$ containing $T$ has non-zero odd vertices. 
Since $\M{C}$ is a minimum trail decomposition on $G(\M{C})$, $T$ is an open trail ending at vertices that are odd in $G(\M{C})$ (Proposition~\ref{prop_mtd_conditions}), which by Lemma~\ref{lem_odd_verts} are also odd in $G$. Therefore by Proposition~\ref{prop_mtd_conditions}, this trail belongs to some minimum trail decomposition of $G$, and hence $\M{C}$ is a subset of some minimum trail decomposition of $G$.

Suppose $G \backslash G(\M{C})$ contains a cycle. Since every vertex is covered by some trail in the trail cover, we could modify a trail that intersects a vertex in the cycle to traverse the entire cycle before returning to the same vertex. This would create a new trail cover $\M{C}^\prime$ where $G(\M{C}) \subsetneq G(\M{C}^\prime)$ and $|\M{C}| = |\M{C}^\prime|$, contradicting the assumption that $\M{C}$ is maximal. Thus $G \backslash G(\M{C})$ must be acyclic.

Now suppose $\M{C}$ is a trail cover that is a subset of some minimum trail decomposition of $G$, the degree of all odd vertices in $G(\M{C})$ is the same as in $G$, and $G \backslash G(\M{C})$ is acyclic. Then naturally $\M{C}$ is a minimum trail decomposition of $G(\M{C})$. Assume there exists another trail cover $\M{C}^\prime$ of $G$ that is a minimum trail decomposition on $G(\M{C}^\prime)$ where $G(\M{C}) \subsetneq G(\M{C}^\prime)$. Then we must show than $|\M{C}^\prime| > |\M{C}|$. 

Since $G \backslash G(\M{C})$ is acyclic, any subgraph is acyclic and hence contains non-zero odd vertices. Since every odd vertex in $G(\M{C})$ is odd in $G$ and has the same degree, the odd vertices in $G(\M{C}^\prime) \backslash G(\M{C})$ are even in $G(\M{C})$. Therefore $G(\M{C}^\prime)$ contains strictly more odd vertices than $G(\M{C})$. Since $\M{C}$ and $\M{C}^\prime$ are both minimum trail decompositions on graphs with at least one odd vertex, we have $|\M{C}| = \frac{1}{2}|\Odd(G(\M{C}))| < \frac{1}{2}|\Odd(G(\M{C}^\prime))| = |\M{C}^\prime|$. Therefore $\M{C}$ is maximal.
\end{proof}

Finally, the following proposition ensures that searching in the space of maximal trail covers is sufficient to solve the minimum trail cover problem.

\begin{restatable}{proposition}{propMTCtoMaximal}
Given a trail cover $\M{C}$ of a graph $G$, we can find a maximal trail cover $\M{C}^\prime$ such that $|\M{C}^\prime| \le |\M{C}|$ in polynomial time.
\end{restatable}
% TODO simplify this. I also think we only need to replace it with a minimum trail decomposition once right?
\begin{proof}
Suppose we are given a trail cover $\M{C}$ of $G$. First replace $\M{C}$ with a minimum trail decomposition on $G(\M{C})$, denoted $\M{C}^\prime$. Note that $|\M{C}^\prime| \le |\M{C}|$. Then for each trail in $\M{C}^\prime$, extend it at both ends to traverse an adjacent edge in $G \backslash G(\M{C}^\prime)$ whenever possible. Continue until no more extensions are possible and we obtain a new trail cover $\M{C}^{\prime\prime}$. For all trails in $\M{C}^{\prime\prime}$ that traverse a vertex belonging to a cycle in $G \backslash G(\M{C}^{\prime\prime})$, modify the trail to traverse this cycle while remaining unchanged elsewhere. Then replace the resulting trail cover with a minimum trail decomposition on $G(\M{C}^{\prime\prime})$ to obtain trail cover $\M{C}^{\prime\prime\prime}$. This trail cover has the property that all trails end at vertices with no adjacent trails in $G \backslash G(\M{C}^{\prime\prime\prime})$. Therefore odd vertices in $G(\M{C}^{\prime\prime\prime})$ are odd in $G$. Since trails in $\M{C}^{\prime\prime\prime}$ belong to some minimum trail decomposition, removing any trail from $G(\M{C}^{\prime\prime\prime})$ produces connected components with non-zero odd vertices, and the same holds when removing the trail from $G$. Therefore trails in $\M{C}^{\prime\prime\prime}$ are part of some minimum trail decomposition in $G$. Since we modified our trail cover to traverse any cycles in the complement graph, $G \backslash G(\M{C}^{\prime\prime\prime})$ is acyclic.

Thus by Theorem~\ref{thm_mtc_subset_mtd}, $\M{C}^{\prime\prime\prime}$ is a maximal trail cover. This constitutes a polynomial reduction since each operation can be performed in polynomial time.
\end{proof}

Therefore, given any minimum trail cover of $G$, we can convert it in polynomial time to a minimum trail cover that is also maximal.

\begin{restatable}{corollary}{propMTCtoMaximalVersion}\label{prop_mtc_mtcwmyf_equivalent}
    The minimum trail cover problem is polynomially equivalent to the problem of finding a minimum trail cover that is maximal.
\end{restatable}

Since every maximal trail cover is a subset of some minimum trail decomposition, we can leverage results and heuristics for minimum trail decompositions when finding minimum trail covers.

\subsection{Approximating minimum trail covers}

From Proposition~\ref{prop_mtc_mtcwmyf_equivalent}, we know it suffices to search in the space of maximal trail covers, which have the structure of minimum trail decompositions.
We can formalize this into a greedy algorithm for finding maximal trail covers of a given graph.

\begin{algorithm}\label{alg_find_tc_min_y}
\ \\\textbf{\textit{Input:}} A graph $G$.\\
\textbf{\textit{Output:}} A maximal trail cover for $G$.
\begin{enumerate}[1.]
\item Search for paths belonging to some minimum trail decomposition using the criteria in Proposition~\ref{prop_mtd_conditions} that traverse only edges with degree at least two. This may be implemented with breadth first search and terminated when the first path is found or a timeout is reached.
\item Remove this path from the graph and continue until no more paths can be removed.
\item Find a minimum trail decomposition of the remaining graph.
\item Return the trail decomposition, which is a maximal trail cover by Theorem~\ref{thm_mtc_subset_mtd}.
\end{enumerate}
\end{algorithm}

This algorithm easily adapts to find $L$-trail covers by subdivision.

\begin{algorithm}\label{alg_find_ltc_min_y}
\ \\\textbf{\textit{Input:}} A graph $G$ and an integer $L$.\\
\textbf{\textit{Output:}} An $L$-trail cover for $G$.
\begin{enumerate}[1.]
\item Execute Steps 1 and 2 of Algorithm~\ref{alg_find_tc_min_y} to remove paths from the graph.
\item Modify Algorithm~\ref{alg_ltrail_decomposition} to attempt to find a trail decomposition of the remaining graph where the length of the trails are of the form $L + k(L + 1)$ for some non-negative integer $k$.
\item Subdivide the trails into $L$-trails and return.
\end{enumerate}
\end{algorithm}

We search for trails with length of the form $L + k(L + 1)$ because when subdividing, we can break the trail into $k$ trails of length $L$ by omitting edges between successive trails and still have a trail cover, as illustrated in the example below where a trail of length five is subdivided into three 1-trails.

\begin{figure}[H]
\centering
\input{diagrams/subdivide_big_trails.tikz}
\end{figure}

Algorithm~\ref{alg_find_tc_min_y} may be executed in polynomial time depending on the search algorithm used, and Algorithm~\ref{alg_find_ltc_min_y} may be executed in polynomial time if the search algorithm used in Algorithm~\ref{alg_ltrail_decomposition} runs in polynomial time. In the worst case, we find no suitable trails with Algorithm~\ref{alg_ltrail_decomposition} and simply return a minimum trail decomposition of size $\frac{1}{2}|\Odd(G)|$.

\begin{remark}
Algorithm~\ref{alg_find_ltc_min_y} is based on subdivision and is therefore easily adapted to the problem of minimising fusions in fusion networks with resource states with bounded photons in the same way as previously outlined in the case of the heuristic algorithm for minimum $L$-trail decompositions.
\end{remark}

\subsection{Reduction to the Travelling Salesman Problem}

The Travelling Salesman Problem (TSP) is one of the most well known and widely studied problems in computer science. The problem is to find a Hamiltonian cycle in a weighted graph that minimises the sum of the edges traversed. Despite being NP-hard, there exist many advanced heuristics and optimisations which allow the TSP to be solved efficiently for graphs containing thousands of vertices~\cite{cite_tsp_benchmarks}.

We will show how the minimum trail cover problem can be reduced to an instance of the graph-theoretic TSP and therefore may leverage existing solvers. We define the \textit{multi-visit TSP} to be a variant of the TSP where we relax the requirement of finding a Hamiltonian cycle to that of finding a trail that visits every vertex in the graph. We then show how the multi-visit TSP can be solved with the original TSP.

\begin{proposition}\label{prop_tsp_alg}
Let $G = (V, E)$ be a graph with non-zero odd vertices. Define $G^\prime = (V^\prime, E^\prime)$ to be the weighted undirected graph obtained by starting with $G$ and adding an addition vertex $v$, $V^\prime = V \cup \{v\}$, additional edges between $v$ and odd vertices of $G$, $E^\prime = E \cup \{(v, w)\ |\ w \in \Odd(G)\}$, and setting the weight of edges in the original graph $G$ to be zero, and the weight of the new edges adjacent to $v$ to be one. Then solutions to the multi-visit TSP on $G^\prime$ correspond to minimum trail covers on $G$ that are maximal.
\end{proposition}

\begin{proof}
Let $T$ be a solution to the multi-visit TSP on $G^\prime$ with total weight $W$. By removing edges from $T$ that are adjacent to the added vertex $v$, we obtain a trail cover on $G \subset G^\prime$ of size $W/2$.

Similarly, given any minimum trail cover that is maximal, we know from Theorem~\ref{thm_mtc_subset_mtd} that trails in $\M{C}$ belong to some minimum trail decomposition and hence each trail ends at distinct odd vertices (Proposition~\ref{prop_mtd_conditions}). Thus we can construct a solution to the multi-visit TSP problem by connecting trails through the added vertex $v$. Since there are $|\M{C}|$ trails and each edge adjacent to $v$ has weight one, the resulting trail has a total weight of $2|\M{C}|$. This is the same weight as was obtained by the solution to the multi-visit TSP and therefore the trail cover obtained from the solution is a minimum trail cover that is maximal.
\end{proof}

This procedure is illustrated in the example below

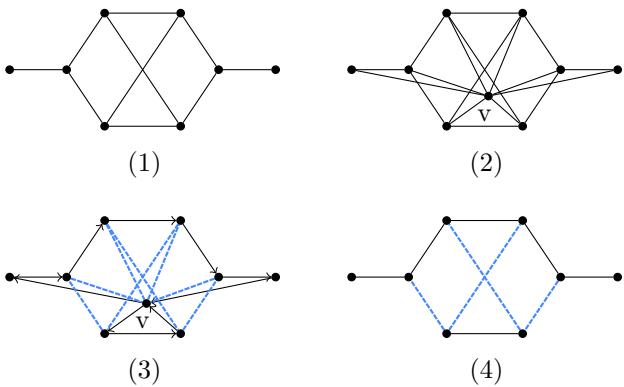
\begin{figure}[H]
\centering
\input{diagrams/travelling_salesman.tikz}
\caption{Finding a minimum trail cover by solving the multi-TSP problem. The original graph (1) is transformed into a new graph (2) by adding a vertex $v$ and edges between $v$ and odd vertices of the original graph. A solution to the multi-visit TSP (3) is transformed into a solution to the TSP on the original graph (4) by removing edges adjacent to $v$. The solution to the TSP on the original graph (4) is a minimum trail cover that is maximal.}
\end{figure}

Note that we can adapt this result to the normal TSP by replacing each node in $G^\prime$ with degree $d > 2$ with a complete subgraph of size $d$. Under this transformation, all trails in the original graph become paths and so solutions to the multi-visit TSP become solutions to the TSP on the new graph.

An example is drawn below where the original graph on the left has a trail that begins and ends at the vertex $s$ and covers every vertex in the graph, however it traverses one vertex twice and so is not a solution to the TSP. The graph on the right has undergone the transformation outlined above and so the solution to the multi-visit TSP now corresponds to a solution to the normal TSP and hence to a minimum trail cover.

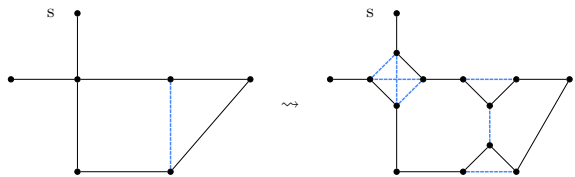
\begin{figure}[H]
\centering
\scalebox{0.7}{\input{diagrams/multi_tsp_to_normal_tsp.tikz}}
\caption{Example of the reduction from the multi-visit TSP to the TSP.}
\end{figure}

The main limitation of this reduction from multi-visit TSP to standard TSP is that the number of edges in the new graph increases with complexity $O(|V|^2 + |E|)$.
Further work could aim at establishing a more efficient reduction that avoids such a large graph expansion and determining whether current TSP solvers can handle the increase efficiently.

%\subsection{Linear programming solution}
%
%Note that we only need to find the subgraph of $G$ corresponding to the minimum trail cover since we can obtain a trail cover by finding a minimum trail decomposition of that subgraph.
%
%Let $e_{ij}$ be zero or one, where a one indicates the presence of the edge in our selection. Here we take $e_{ij} = e_{ji}$ since its an undirected graph. Then the constraint is simply that every vertex is covered by at least one edge.
%
%\[
%\forall i \in V \sum_{j\in N(i)} e_{ij} \ge 1.
%\]
%
%With the additional variables $x_i$ for $i \in \Odd(G)$ as having the constraints
%
%\[
%2\sum_{j \in N(i)} e_{ij} - 4 \le 4x_i \le 2(\sum_{j \in N(i)} e_{ij}) - 7
%\]
%
%Which is essentially to achieve $x_i = \lceil (\sum_{j \in N(i)} e_{ij} / 2) \rceil - 1$
%
%The number of Y fusions is just $|E| - \sum_{i, j} e_{ij}$ where the sum iterates over each edge once. The number of $X$ fusions is more tricky.

%% file: diagrams/example_of_reaching_bound.tikz
\begin{tikzpicture}
	\begin{pgfonlayer}{nodelayer}
		\node [style=vertex] (0) at (0.85, 2) {};
		\node [style=vertex] (1) at (1.6, 2) {};
		\node [style=vertex] (2) at (2.35, 2) {};
		\node [style=vertex] (3) at (3.1, 2) {};
		\node [style=vertex] (4) at (1.6, 2.75) {};
		\node [style=vertex] (5) at (2.35, 2.75) {};
		\node [style=vertex] (6) at (4.375, 2) {};
		\node [style=vertex] (7) at (5.125, 2) {};
		\node [style=vertex] (8) at (5.875, 2) {};
		\node [style=vertex] (9) at (6.625, 2) {};
		\node [style=vertex] (10) at (5.125, 3) {};
		\node [style=vertex] (11) at (5.875, 3) {};
		\node [style=vertex] (12) at (5.125, 2.5) {};
		\node [style=vertex] (13) at (5.875, 2.5) {};
		\node [style=vertex] (14) at (0.75, -0.25) {};
		\node [style=vertex] (15) at (1.5, -0.25) {};
		\node [style=vertex] (16) at (2, -0.25) {};
		\node [style=vertex] (17) at (3.25, -0.25) {};
		\node [style=vertex] (18) at (1.5, 0.75) {};
		\node [style=vertex] (19) at (2.5, 0.75) {};
		\node [style=vertex] (20) at (1.5, 0.25) {};
		\node [style=vertex] (21) at (2.5, 0.25) {};
		\node [style=vertex] (22) at (2.5, -0.25) {};
		\node [style=vertex] (23) at (4.25, -0.25) {};
		\node [style=vertex] (24) at (5, -0.25) {};
		\node [style=vertex] (25) at (6, -0.25) {};
		\node [style=vertex] (26) at (6.5, -0.25) {};
		\node [style=vertex] (27) at (5, 0.75) {};
		\node [style=vertex] (28) at (6, 0.75) {};
		\node [style=vertex] (29) at (5.5, 0.75) {};
		\node [style=vertex] (30) at (6, 0.25) {};
		\node [style=vertex] (31) at (5.5, -0.25) {};
		\node [style=none] (32) at (5.5, -0.75) {D};
		\node [style=none] (33) at (2, -0.75) {C};
		\node [style=none] (34) at (2, 1.5) {A};
		\node [style=none] (35) at (5.5, 1.5) {B};
		\node [style=vertex] (37) at (2, 0.75) {};
	\end{pgfonlayer}
	\begin{pgfonlayer}{edgelayer}
		\draw (0) to (1);
		\draw (1) to (2);
		\draw (2) to (3);
		\draw (5) to (4);
		\draw (4) to (1);
		\draw (5) to (2);
		\draw (6) to (7);
		\draw (7) to (8);
		\draw (8) to (9);
		\draw (11) to (10);
		\draw (13) to (11);
		\draw (12) to (10);
		\draw (14) to (15);
		\draw (15) to (16);
		\draw (21) to (19);
		\draw (20) to (18);
		\draw (22) to (17);
		\draw (23) to (24);
		\draw (25) to (26);
		\draw (27) to (24);
		\draw (29) to (28);
		\draw (30) to (28);
		\draw (31) to (25);
		\draw (37) to (18);
		\draw [style=X fusion] (12) to (7);
		\draw [style=X fusion] (13) to (8);
		\draw [style=X fusion] (20) to (15);
		\draw [style=X fusion] (21) to (22);
		\draw [style=X fusion] (22) to (16);
		\draw [style=X fusion] (37) to (19);
		\draw [style=X fusion] (27) to (29);
		\draw [style=X fusion] (30) to (25);
		\draw [style=X fusion] (31) to (24);
	\end{pgfonlayer}
\end{tikzpicture}

%% file: diagrams/joining_closed_trails.tikz
\begin{tikzpicture}
	\begin{pgfonlayer}{nodelayer}
		\node [style=vertex] (3) at (-3, -1.25) {};
		\node [style=vertex] (4) at (-3, -0.75) {};
		\node [style=vertex] (7) at (-2.5, -0.25) {};
		\node [style=vertex] (11) at (-3.5, -0.25) {};
		\node [style=vertex] (16) at (-1.75, -1.25) {};
		\node [style=vertex] (19) at (-1.75, -2.25) {};
		\node [style=vertex] (20) at (-2, -2.5) {};
		\node [style=vertex] (23) at (-3, -2.5) {};
		\node [style=vertex] (30) at (1.5, -0.75) {};
		\node [style=vertex] (33) at (0, -0.75) {};
		\node [style=vertex] (36) at (1.75, -1.25) {};
		\node [style=vertex] (40) at (1.75, -2.5) {};
		\node [style=vertex] (43) at (0.5, -2.5) {};
		\node [style=vertex] (44) at (1, -1.25) {};
		\node [style=vertex] (45) at (0.5, -1.25) {};
		\node [style=none] (46) at (-0.75, -1.25) {$\rightsquigarrow$};
		\node [style=vertex] (49) at (-2.5, 0.5) {};
		\node [style=vertex] (50) at (-3, 0.5) {};
		\node [style=vertex] (53) at (-3.75, 0.5) {};
		\node [style=vertex] (56) at (-1.75, 0.5) {};
		\node [style=vertex] (60) at (0.75, 0.5) {};
		\node [style=vertex] (63) at (0, 0.5) {};
		\node [style=vertex] (66) at (1.5, 0.5) {};
		\node [style=none] (67) at (-0.75, 0.5) {$\rightsquigarrow$};
	\end{pgfonlayer}
	\begin{pgfonlayer}{edgelayer}
		\draw (23) to (3);
		\draw (3) to (16);
		\draw (19) to (16);
		\draw (23) to (20);
		\draw (7) to (4);
		\draw (11) to (4);
		\draw (43) to (40);
		\draw (30) to (44);
		\draw (44) to (36);
		\draw (40) to (36);
		\draw (43) to (45);
		\draw (45) to (33);
		\draw (53) to (50);
		\draw (49) to (56);
		\draw (63) to (60);
		\draw (60) to (66);
		\draw [style=X fusion] (45) to (44);
		\draw [style=X fusion] (4) to (3);
		\draw [style=X fusion] (19) to (20);
		\draw [style=X fusion] (49) to (50);
	\end{pgfonlayer}
\end{tikzpicture}

%% file: diagrams/subdivide_big_trails.tikz
\begin{tikzpicture}
	\begin{pgfonlayer}{nodelayer}
		\node [style=vertex] (0) at (1.25, 0) {};
		\node [style=vertex] (1) at (2, 0) {};
		\node [style=vertex] (2) at (2.75, 0) {};
		\node [style=vertex] (3) at (3.5, 0) {};
		\node [style=vertex] (4) at (4.25, 0) {};
		\node [style=vertex] (5) at (5, 0) {};
		\node [style=none] (6) at (3.1, -0.5) {\rotatebox[origin=c]{270}{$\rightsquigarrow$}};
		\node [style=vertex] (7) at (1.25, -1) {};
		\node [style=vertex] (8) at (2, -1) {};
		\node [style=vertex] (9) at (2.75, -1) {};
		\node [style=vertex] (10) at (3.5, -1) {};
		\node [style=vertex] (11) at (4.25, -1) {};
		\node [style=vertex] (12) at (5, -1) {};
	\end{pgfonlayer}
	\begin{pgfonlayer}{edgelayer}
		\draw (0) to (1);
		\draw (1) to (2);
		\draw (2) to (3);
		\draw (3) to (4);
		\draw (4) to (5);
		\draw (7) to (8);
		\draw (9) to (10);
		\draw (11) to (12);
	\end{pgfonlayer}
\end{tikzpicture}

%% file: diagrams/travelling_salesman.tikz
\begin{tikzpicture}
	\begin{pgfonlayer}{nodelayer}
		\node [style=vertex] (38) at (12.75, -4.25) {};
		\node [style=vertex] (39) at (13.75, -4.25) {};
		\node [style=vertex] (40) at (14.25, -5) {};
		\node [style=vertex] (41) at (13.75, -5.75) {};
		\node [style=vertex] (42) at (12.75, -5.75) {};
		\node [style=vertex] (43) at (12.25, -5) {};
		\node [style=vertex] (44) at (11.5, -5) {};
		\node [style=vertex] (45) at (15, -5) {};
		\node [style=vertex] (46) at (13.3, -5.35) {};
		\node [style=vertex] (47) at (8.25, -7) {};
		\node [style=vertex] (48) at (9.25, -7) {};
		\node [style=vertex] (49) at (9.75, -7.75) {};
		\node [style=vertex] (50) at (9.25, -8.5) {};
		\node [style=vertex] (51) at (8.25, -8.5) {};
		\node [style=vertex] (52) at (7.75, -7.75) {};
		\node [style=vertex] (53) at (7, -7.75) {};
		\node [style=vertex] (54) at (10.5, -7.75) {};
		\node [style=vertex] (55) at (8.8, -8.1) {};
		\node [style=none] (56) at (13.25, -5.6) {v};
		\node [style=none] (57) at (8.75, -8.35) {v};
		\node [style=vertex] (58) at (12.75, -7) {};
		\node [style=vertex] (59) at (13.75, -7) {};
		\node [style=vertex] (60) at (14.25, -7.75) {};
		\node [style=vertex] (61) at (13.75, -8.5) {};
		\node [style=vertex] (62) at (12.75, -8.5) {};
		\node [style=vertex] (63) at (12.25, -7.75) {};
		\node [style=vertex] (64) at (11.5, -7.75) {};
		\node [style=vertex] (65) at (15, -7.75) {};
		\node [style=vertex] (68) at (8.25, -4.25) {};
		\node [style=vertex] (69) at (9.25, -4.25) {};
		\node [style=vertex] (70) at (9.75, -5) {};
		\node [style=vertex] (71) at (9.25, -5.75) {};
		\node [style=vertex] (72) at (8.25, -5.75) {};
		\node [style=vertex] (73) at (7.75, -5) {};
		\node [style=vertex] (74) at (7, -5) {};
		\node [style=vertex] (75) at (10.5, -5) {};
		\node [style=none] (77) at (8.775, -6.25) {(1)};
		\node [style=none] (78) at (13.275, -6.25) {(2)};
		\node [style=none] (79) at (8.775, -9) {(3)};
		\node [style=none] (80) at (13.275, -9) {(4)};
	\end{pgfonlayer}
	\begin{pgfonlayer}{edgelayer}
		\draw (42) to (41);
		\draw (40) to (39);
		\draw (39) to (38);
		\draw (38) to (43);
		\draw (45) to (40);
		\draw (44) to (43);
		\draw (46) to (41);
		\draw (46) to (45);
		\draw (46) to (39);
		\draw (46) to (38);
		\draw (46) to (43);
		\draw (46) to (40);
		\draw (46) to (44);
		\draw (46) to (42);
		\draw (40) to (41);
		\draw (42) to (39);
		\draw (41) to (38);
		\draw (43) to (42);
		\draw (62) to (61);
		\draw (60) to (59);
		\draw (59) to (58);
		\draw (58) to (63);
		\draw (65) to (60);
		\draw (64) to (63);
		\draw [style=diredge] (55) to (51);
		\draw [style=diredge] (51) to (50);
		\draw [style=diredge] (50) to (55);
		\draw [style=diredge] (55) to (53);
		\draw [style=diredge] (53) to (52);
		\draw [style=diredge] (52) to (47);
		\draw [style=diredge] (47) to (48);
		\draw [style=diredge] (48) to (49);
		\draw [style=diredge] (49) to (54);
		\draw [style=diredge] (54) to (55);
		\draw (72) to (71);
		\draw (70) to (69);
		\draw (69) to (68);
		\draw (68) to (73);
		\draw (75) to (70);
		\draw (74) to (73);
		\draw (70) to (71);
		\draw (72) to (69);
		\draw (71) to (68);
		\draw (73) to (72);
		\draw [style=Y fusion] (61) to (58);
		\draw [style=Y fusion] (60) to (61);
		\draw [style=Y fusion] (62) to (63);
		\draw [style=Y fusion] (62) to (59);
		\draw [style=Y fusion] (50) to (49);
		\draw [style=Y fusion] (51) to (52);
		\draw [style=Y fusion] (52) to (55);
		\draw [style=Y fusion] (47) to (55);
		\draw [style=Y fusion] (55) to (48);
		\draw [style=Y fusion] (48) to (51);
		\draw [style=Y fusion] (47) to (50);
		\draw [style=hadamard edge] (49) to (55);
	\end{pgfonlayer}
\end{tikzpicture}

%% file: diagrams/multi_tsp_to_normal_tsp.tikz
\begin{tikzpicture}
	\begin{pgfonlayer}{nodelayer}
		\node [style=none] (39) at (22.75, 0.5) {$\rightsquigarrow$};
		\node [style=vertex] (64) at (17.5, 1) {};
		\node [style=vertex] (65) at (18.75, 2.25) {};
		\node [style=vertex] (66) at (18.75, 1) {};
		\node [style=vertex] (67) at (18.75, -0.75) {};
		\node [style=vertex] (68) at (20.5, -0.75) {};
		\node [style=vertex] (69) at (20.5, 1) {};
		\node [style=vertex] (70) at (22, 1) {};
		\node [style=vertex] (71) at (23.5, 1) {};
		\node [style=vertex] (72) at (24.75, 2.25) {};
		\node [style=vertex] (74) at (24.75, -0.75) {};
		\node [style=vertex] (75) at (26, -0.75) {};
		\node [style=vertex] (77) at (28, 1) {};
		\node [style=vertex] (78) at (24.75, 1.5) {};
		\node [style=vertex] (79) at (24.75, 0.5) {};
		\node [style=vertex] (80) at (25.25, 1) {};
		\node [style=vertex] (81) at (24.25, 1) {};
		\node [style=vertex] (82) at (26, 1) {};
		\node [style=vertex] (83) at (26.5, 0.5) {};
		\node [style=vertex] (84) at (27, 1) {};
		\node [style=vertex] (85) at (26.5, -0.25) {};
		\node [style=vertex] (86) at (27, -0.75) {};
		\node [style=none] (87) at (18.25, 2.25) {s};
		\node [style=none] (88) at (24.25, 2.25) {s};
	\end{pgfonlayer}
	\begin{pgfonlayer}{edgelayer}
		\draw (70) to (69);
		\draw (68) to (70);
		\draw (69) to (66);
		\draw (66) to (65);
		\draw (64) to (66);
		\draw (66) to (67);
		\draw (67) to (68);
		\draw (74) to (79);
		\draw (79) to (81);
		\draw (81) to (71);
		\draw (72) to (78);
		\draw (78) to (80);
		\draw (83) to (84);
		\draw (82) to (83);
		\draw (82) to (80);
		\draw (84) to (77);
		\draw (86) to (77);
		\draw (86) to (85);
		\draw (85) to (75);
		\draw (75) to (74);
		\draw [style=Y fusion] (68) to (69);
		\draw [style=Y fusion] (79) to (80);
		\draw [style=Y fusion] (80) to (81);
		\draw [style=Y fusion] (81) to (78);
		\draw [style=Y fusion] (78) to (79);
		\draw [style=Y fusion] (82) to (84);
		\draw [style=Y fusion] (75) to (86);
		\draw [style=Y fusion] (85) to (83);
	\end{pgfonlayer}
\end{tikzpicture}

%% file: 5-graphrewrites.tex
In this section, we use results from previous sections to develop graph rewrite strategies that transform graph states to reduce the required number of fusions and enhance the performance of approximation algorithms developed in Section~\ref{sec:approximation}. The algorithms presented here typically reduce the number of edges in the graphs thus provide utility for fusion networks created with star resource states as well, though we note that the algorithms presented in~\cite{cite_cz_complexity} would be more effective in this case.

From Section~\ref{sec:complexity}, we know that without allowing graph rewrites, finding linear XY-fusion networks with the minimum number of fusions is NP-hard in all cases except for X fusions with unbounded resource states. Whether these problems remain NP-hard when allowing rewrites remains unknown.

One might ask whether we can transform any graph into one with bounded tree-width to solve trail cover problems in polynomial time. However, this is not possible since tree-width is bounded below by rank-width~\cite{cite_cz_complexity}, and because rank-width is invariant under local complementation and performing inverse Z-deletions never decreases rank width~\cite{cite_rank_width}, we cannot reduce tree-width arbitrarily.

In order for the open graph to be deterministically implementable, we require our rewrites preserve gflow in the graph. Backens and McElvanney\cite{pauli_preserving_rewrites} showed that local complementation and Z-deletion (and their inverses) not only preserve gflow, but are sufficient to transform any two labeled open graphs with gflow and the same target linear map into each other, and so our rewrites will exclusively use these two rewrites.

Our main optimization exploits the fact that the number of fusions required to implement a graph state corresponding to graph $G = (V, E)$ with trail cover $\M{C}$ is given by
\begin{equation}\label{eq_num_fusions}
|E| - |V| + |\M{C}|.
\end{equation}

From~\eqref{eq_num_fusions}, we can reduce the number of fusions in a linear XY-fusion network by either decreasing $|E|$ and $|\M{C}|$ or increasing $|V|$.

\subsection{Rewrite 1: Reducing edges}

Our first rewrite strategy is a known heuristic for simplifying graph states: apply local complementation whenever it reduces the number of edges in the graph.

\textit{Rewrite 1: Locally complement a vertex whenever it reduces the number of edges in the graph.}

An example of such a situation is illustrated in Figure~\ref{fig_optimal_lc}.

\begin{figure}[h]
    \centering
    \input{diagrams/optimal_lc_path.tikz}
    \caption{The cycle graph state is transformed into a linear graph state requiring fewer fusions through a series of local complementation. However, any local complementation applied to the cycle will increase the number of edges in the graph state and so Rewrite 1 would not follow the optimal rewrites steps above, and therefore does not always find the optimal graph state.}
    \label{fig_optimal_lc}
\end{figure}
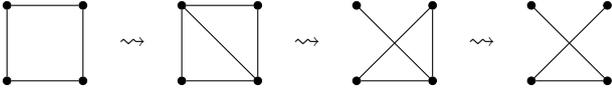

%To analyse the performance on larger graphs, we ran full reduce on randomly generated graphs with different levels of sparsity. A graph with edge density $p$ means that each edge has probability $p$ of being present in the graph. We repeated the experiment for graphs with 20 vertices and 50 vertices.
%
%TODO combine these graphs
%\begin{figure}[H]
%\begin{subfigure}[h]{0.5\linewidth}
%% sh run.sh 20 100 greedy
%\includegraphics[width=\linewidth]{images/edge_reduction_20_greedy.png}
%\end{subfigure}
%\begin{subfigure}[h]{0.5\linewidth}
%% sh run.sh 50 100 greedy
%\includegraphics[width=\linewidth]{images/edge_reduction_50_greedy.png}
%\end{subfigure}
%\caption{Average reduction in the number of edges after applying Rewrite 1.}
%\end{figure}
%
%It is clear that the algorithm is most effective on dense graphs. This is because local complementations are only performed when the number of edges strictly decreases. In fact, using this simple procedure any graph can be reduced to one with density less than 0.5. For dense graphs, we see a much greater reduction in the number of edges due to tendency for local complementations to significantly reduce the number of edges. This is most easily seen in the case of complete graphs where locally complementing any vertex will transform the graph into a star graph --- the most sparse connected graph.

We experimentally verified that greedily locally complementing vertices whenever it decreases the total number of edges never increases the minimum path cover size for graphs with fewer than 9 vertices. However, a counterexample exists for graphs with 9 vertices, suggesting this may cease to be an effective heuristic beyond a certain graph size.

\subsection{Rewrite 2: Complementing cliques}

In general, Z-deletions increase the number of edges in the graph and hence the number of fusions. 
However, there is a special case which we can show will never increase the number of Y fusions required in a linear Y-fusion network.

\begin{definition}
Given a graph with a clique, \textit{complementing the clique} refers to using an inverse Z-deletion to add a vertex connected to the clique vertices, then performing local complementation on that vertex.
\end{definition}

\textit{Rewrite 2: Complement every clique in the graph.}

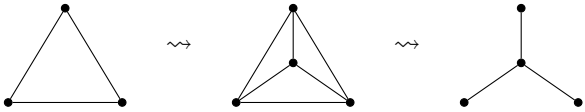
\begin{figure}[H]
\centering
\input{diagrams/comp-triangle.tikz}
\caption{Complementing a triangle by first adding a new node connected to the triangle vertices and locally complementing it.}
\label{fig:comp-triangle}
\end{figure}

Complementing the triangle keeps the number of edges unchanged but increases the number of vertices by one. Therefore by~\eqref{eq_num_fusions}, it decreases the number of fusions if the path cover stays the same size or increases by one. We now show this is always the case.

\begin{proposition}
Complementing a triangle in $G$ never increases the size of a minimum path cover of $G$.
\end{proposition}

\begin{proof}
Assume we have a triangle and a minimum path cover $P$. We show that we can always update $P$ after complementing the triangle such that the number of fusions never increases.

Consider the cases where the triangle contains either 0, 1, or 2 edges from the cover. It cannot contain 3 edges as this would imply the paths are not vertex disjoint. These cases are illustrated in Figure~\ref{fig_comp-tri} alongside a modified path flowing through the triangle which does not increase the number of fusions.
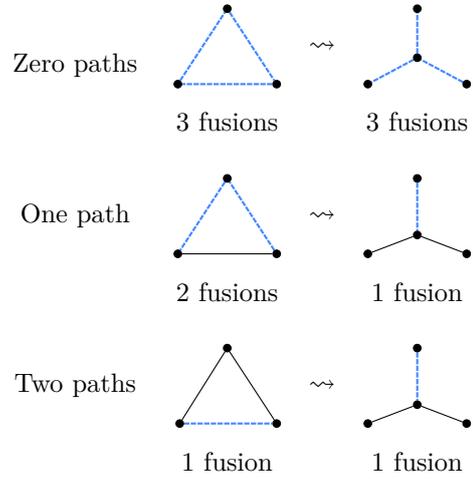
\begin{figure}[H]
\centering
\input{diagrams/complementing_triangles_good.tikz}
\caption{Complementing the triangle never increases the number of fusions. Solid black lines represent edges belonging to the path cover and blue dashed lines are $Y$ fusions.}
\label{fig_comp-tri}
\end{figure}
\vspace{-5mm}
\end{proof}

This technique easily generalizes to cliques of arbitrary size. 
For cliques of size four or greater, it is evident that the number of fusions does not increase since when $n = 4$, the number of edges reduces by 2 and the number of vertices increases by 1. 
However, note that edges in a 4-clique could be traversed by at most two paths, so complementing the clique will at most split one path into two and increase the path cover size by one. Therefore from~\eqref{eq_num_fusions}, we always strictly decrease the number of fusions by complementing cliques of size four or greater.

\begin{proposition}
    Complementing a clique in $G$ never increases the number of fusions required to implement $G$.
\end{proposition}

By complementing $k$-cliques when $k > 3$, we also reduce the number of edges, which may reduce the time required to find the minimum path cover, 
though this comes at the cost of increasing the number of nodes.
We note moreover that this result may not hold generally for bounded path covers.

\subsection{Rewrite 3: Reducing odd vertices}

For $X$ fusions, we have the advantage of knowing the size of the minimum trail decomposition will be half the odd number of vertices. Therefore, in addition to reducing the number of edges, we also want to reduce the number of odd vertices. For bounded trail decompositions, reducing the number of odd vertices also improves the approximation ratio of Algorithm~\ref{alg_ltrail_decomposition}. 

\textit{Rewrite 3: Perform local complementation whenever it reduces $|E| + \frac{1}{2}|\Odd(G)|$.}

Figure~\ref{fig_full_reduce_vs_lc} illustrates a graph where local complementation of the vertex $v$ leaves the number of edges unchanged but reduces the number of odd edges by four. The new graph can therefore be implemented with two X fusions, whereas the original graph requires at least four $XY$ fusions.

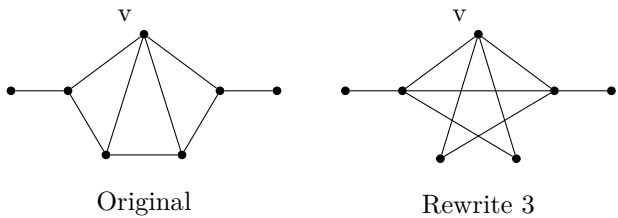
\begin{figure}[H]
\centering
\input{diagrams/full_reduce_vs_lc.tikz}
\caption{The diagram on the left remains unchanged by PyZX's full reduce algorithm and the minimum fusion network requires six $X$ fusions. The diagram on the right is the output of Rewrite 3 and requires only two $X$ fusions.} 
\label{fig_full_reduce_vs_lc}
\end{figure}

\begin{proposition}
Complementing a triangle reduces $X$ fusions if and only if two or more vertices in the triangle are odd.
\end{proposition}

\begin{proof}
Let $G = (V, E)$ be a graph and $S \subset G$ be a triangle in $G$ and let $F$ be the number of $X$ fusions required to implement $G$. After complementing the triangle, we obtain the graph $G^\prime = (V^\prime, E^\prime)$ where $|E^\prime| = |E|$, $|V^\prime| = |V| + 1$, and the odd number of vertices increases by $4 - 2|\Odd(S)|$. If $\M{T}$ is a minimum trail decomposition of $G$, Theorem~\ref{thm_num_fusions} tells us the required number of fusions $F^\prime$ to implement $G^\prime$ increases by
\begin{align*}
|F^\prime| - |F| &= |E^\prime| - |V^\prime| + |\M{T}^\prime| - (|E| - |V| + |\M{T}|) \\
&= |\M{T}^\prime| - |\M{T}| - 1\\
&= \frac{1}{2}\left[ |\Odd{G}| + 4 - 2|\Odd{S}|\right] - \frac{1}{2}|\Odd{G}| - 1 \\
&= 1 - |\Odd{S}|
\end{align*}
Therefore the required number of fusions decreases if and only if the triangle contains two or more odd vertices.
\end{proof}
Note that this does not hold in general for bounded trail decompositions; a counterexample appears in Appendix~\ref{appendix_sec_tri_complementation_counterexample}.

%% file: diagrams/optimal_lc_path.tikz
\begin{tikzpicture}
	\begin{pgfonlayer}{nodelayer}
		\node [style=vertex] (0) at (-3, 0) {};
		\node [style=vertex] (1) at (-2, 0) {};
		\node [style=vertex] (2) at (-2, 1) {};
		\node [style=vertex] (3) at (-3, 1) {};
		\node [style=none] (4) at (-1.35, 0.5) {$\rightsquigarrow$};
		\node [style=vertex] (5) at (-0.7, 0) {};
		\node [style=vertex] (6) at (0.3, 0) {};
		\node [style=vertex] (7) at (0.3, 1) {};
		\node [style=vertex] (8) at (-0.7, 1) {};
		\node [style=none] (9) at (0.95, 0.5) {$\rightsquigarrow$};
		\node [style=vertex] (10) at (1.6, 0) {};
		\node [style=vertex] (11) at (2.6, 0) {};
		\node [style=vertex] (12) at (2.6, 1) {};
		\node [style=vertex] (13) at (1.6, 1) {};
		\node [style=none] (14) at (3.25, 0.5) {$\rightsquigarrow$};
		\node [style=vertex] (15) at (3.9, 0) {};
		\node [style=vertex] (16) at (4.9, 0) {};
		\node [style=vertex] (17) at (4.9, 1) {};
		\node [style=vertex] (18) at (3.9, 1) {};
	\end{pgfonlayer}
	\begin{pgfonlayer}{edgelayer}
		\draw (2) to (1);
		\draw (1) to (0);
		\draw (0) to (3);
		\draw (3) to (2);
		\draw (7) to (6);
		\draw (6) to (5);
		\draw (5) to (8);
		\draw (8) to (7);
		\draw (8) to (6);
		\draw (12) to (11);
		\draw (11) to (10);
		\draw (13) to (11);
		\draw (10) to (12);
		\draw (16) to (15);
		\draw (18) to (16);
		\draw (15) to (17);
	\end{pgfonlayer}
\end{tikzpicture}

%% file: diagrams/comp-triangle.tikz
\begin{tikzpicture}
	\begin{pgfonlayer}{nodelayer}
		\node [style=vertex] (0) at (-2.5, 0) {};
		\node [style=vertex] (1) at (-1.75, 1.25) {};
		\node [style=vertex] (2) at (-1, 0) {};
		\node [style=vertex] (9) at (0.5, 0) {};
		\node [style=vertex] (10) at (1.25, 1.25) {};
		\node [style=vertex] (11) at (2, 0) {};
		\node [style=vertex] (18) at (1.25, 0.525) {};
		\node [style=vertex] (19) at (3.5, 0) {};
		\node [style=vertex] (20) at (4.25, 1.25) {};
		\node [style=vertex] (21) at (5, 0) {};
		\node [style=vertex] (28) at (4.25, 0.525) {};
		\node [style=none] (29) at (-0.25, 0.75) {$\rightsquigarrow$};
		\node [style=none] (30) at (2.75, 0.75) {$\rightsquigarrow$};
	\end{pgfonlayer}
	\begin{pgfonlayer}{edgelayer}
		\draw (2) to (0);
		\draw (0) to (1);
		\draw (1) to (2);
		\draw (11) to (9);
		\draw (9) to (10);
		\draw (10) to (11);
		\draw (18) to (9);
		\draw (18) to (11);
		\draw (18) to (10);
		\draw (28) to (19);
		\draw (28) to (21);
		\draw (28) to (20);
	\end{pgfonlayer}
\end{tikzpicture}

%% file: diagrams/complementing_triangles_good.tikz
\begin{tikzpicture}
	\begin{pgfonlayer}{nodelayer}
		\node [style=none] (7) at (0, 3.75) {$\rightsquigarrow$};
		\node [style=none] (15) at (0, 1.5) {$\rightsquigarrow$};
		\node [style=vertex] (24) at (-1.25, 4.25) {};
		\node [style=vertex] (25) at (-1.9, 3.25) {};
		\node [style=vertex] (26) at (-0.6, 3.25) {};
		\node [style=vertex] (27) at (1.25, 4.25) {};
		\node [style=vertex] (28) at (0.6, 3.25) {};
		\node [style=vertex] (29) at (1.9, 3.25) {};
		\node [style=vertex] (30) at (-1.25, 2) {};
		\node [style=vertex] (31) at (-1.9, 1) {};
		\node [style=vertex] (32) at (-0.6, 1) {};
		\node [style=vertex] (33) at (1.25, 2) {};
		\node [style=vertex] (34) at (0.6, 1) {};
		\node [style=vertex] (35) at (1.9, 1) {};
		\node [style=vertex] (36) at (-1.25, -0.25) {};
		\node [style=vertex] (37) at (-1.875, -1.25) {};
		\node [style=vertex] (38) at (-0.6, -1.25) {};
		\node [style=vertex] (39) at (1.25, 3.6) {};
		\node [style=vertex] (40) at (1.25, 1.25) {};
		\node [style=vertex] (41) at (1.25, -1) {};
		\node [style=vertex] (42) at (1.25, -0.25) {};
		\node [style=vertex] (43) at (0.6, -1.25) {};
		\node [style=vertex] (44) at (1.9, -1.25) {};
		\node [style=none] (45) at (0, -0.75) {$\rightsquigarrow$};
		\node [style=none] (46) at (-3.25, 3.5) {\textrm{Zero paths}};
		\node [style=none] (47) at (-3.25, 1.5) {\textrm{One path}};
		\node [style=none] (48) at (-3.25, -0.75) {\textrm{Two paths}};
		\node [style=none] (49) at (-1.25, 2.75) {\tx{3 fusions}};
		\node [style=none] (50) at (1.25, 2.75) {\tx{3 fusions}};
		\node [style=none] (52) at (-1.25, 0.5) {\tx{2 fusions}};
		\node [style=none] (53) at (1.25, 0.5) {\tx{1 fusion}};
		\node [style=none] (55) at (-1.25, -1.75) {\tx{1 fusion}};
		\node [style=none] (56) at (1.25, -1.75) {\tx{1 fusion}};
	\end{pgfonlayer}
	\begin{pgfonlayer}{edgelayer}
		\draw [style=none] (31) to (32);
		\draw [style=none] (34) to (40);
		\draw [style=none] (40) to (35);
		\draw [style=none] (37) to (36);
		\draw [style=none] (36) to (38);
		\draw [style=none] (41) to (43);
		\draw [style=none] (41) to (44);
		\draw [style=Y fusion] (29) to (39);
		\draw [style=Y fusion] (39) to (28);
		\draw [style=Y fusion] (39) to (27);
		\draw [style=Y fusion] (26) to (25);
		\draw [style=Y fusion] (25) to (24);
		\draw [style=Y fusion] (24) to (26);
		\draw [style=Y fusion] (32) to (30);
		\draw [style=Y fusion] (30) to (31);
		\draw [style=Y fusion] (33) to (40);
		\draw [style=Y fusion] (37) to (38);
		\draw [style=Y fusion] (42) to (41);
	\end{pgfonlayer}
\end{tikzpicture}

%% file: diagrams/full_reduce_vs_lc.tikz
\begin{tikzpicture}
	\begin{pgfonlayer}{nodelayer}
		\node [style=vertex] (0) at (-1, 2) {};
		\node [style=vertex] (1) at (0, 2.75) {};
		\node [style=vertex] (2) at (1, 2) {};
		\node [style=vertex] (3) at (1.75, 2) {};
		\node [style=vertex] (4) at (-1.75, 2) {};
		\node [style=vertex] (5) at (-0.5, 1.15) {};
		\node [style=vertex] (6) at (0.5, 1.15) {};
		\node [style=none] (7) at (0, 0.5) {Original};
		\node [style=vertex] (8) at (3.4, 2) {};
		\node [style=vertex] (9) at (4.4, 2.75) {};
		\node [style=vertex] (10) at (5.4, 2) {};
		\node [style=vertex] (11) at (6.15, 2) {};
		\node [style=vertex] (12) at (2.65, 2) {};
		\node [style=vertex] (13) at (3.9, 1.1) {};
		\node [style=vertex] (14) at (4.9, 1.1) {};
		\node [style=none] (15) at (4.4, 0.5) {Rewrite 3};
		\node [style=none] (16) at (4.15, 3) {v};
		\node [style=none] (17) at (-0.25, 3) {v};
	\end{pgfonlayer}
	\begin{pgfonlayer}{edgelayer}
		\draw (0) to (1);
		\draw (1) to (2);
		\draw (2) to (6);
		\draw (6) to (5);
		\draw (5) to (0);
		\draw (0) to (4);
		\draw (2) to (3);
		\draw (1) to (5);
		\draw (1) to (6);
		\draw (8) to (9);
		\draw (9) to (10);
		\draw (8) to (12);
		\draw (10) to (11);
		\draw (9) to (13);
		\draw (9) to (14);
		\draw (14) to (8);
		\draw (13) to (10);
		\draw (8) to (10);
	\end{pgfonlayer}
\end{tikzpicture}

%% file: 6-benchmarks.tex
We now present a complete compilation algorithm that integrates the approximation algorithms from Section~\ref{sec:approximation} with the graph rewrites from Section~\ref{sec:graphrewrites} to compile MBQC patterns into optimized fusion networks for different fusion types. 
We evaluate the performance of our compilation algorithms on two distinct benchmark sets. 
The first set, denoted B1, is comprised of common quantum error-correcting codes: star graphs~\cite{cite_ft_quantum_computing_with_nondeterminism, cite_photonic_architecture_with_ghz} used in the original FBQC formulation~\cite{cite_orig_fbqc}, cycles~\cite{cite_orig_fbqc}, lattices~\cite{cite_lattice_mbqc}, repeaters~\cite{cite_repeaters}, and trees~\cite{cite_trees_mbqc, cite_loss_tolerance_trees}.
The second benchmark set, B2, is a subset of the circuits from the QASMBench~\cite{cite_qasm_bench} suite, which contains real-world quantum algorithms including Grover's algorithm, quantum teleportation, and the Quantum Fourier Transform. 
We converted the OpenQASM circuits to MBQC patterns using PyZX~\cite{cite_pyzx}, transforming each circuit into an open graph from which we extracted the graph structure. This process includes applying the ``full reduce'' graph state reduction procedure that partially simplifies the graph.

To evaluate the quality of the compiled fusion networks, we analyse the number of fusions, number of photons and the number of resource states required to implement a given graph, which are linearly related by Theorem~\ref{thm_num_fusions}.
We compare the number of fusions against a (loose) lower bound which assumes that the graph is generated from a single resource state.
Since deterministic resource state generation is impractical for all but the smallest states, we address this limitation by bounding the allowed number of photons in each resource state.
We conclude by analysing the probability of successful graph state generation in a photonic architecture comprising caterpillar state sources~\cite{cite_deterministic_catapillar_states, cite_linear_cluster_state_generation}, 
unrestricted routers and repeat-until-success fusion modules \cite{cite_rus_original,lee_nearly_2015, hilaire_enhanced_2024, felice_fusion_2024}.

\begin{figure}[h]
    \centering
    \input{diagrams/test_graphs_no_zx.tikz}
\end{figure}

\subsection{Performance of approximation algorithms}

We compute fusion networks using the following approaches:
\textbf{X fusion networks:} We follow the constructive proof of Theorem~\ref{thm_min_trail_decomp}, utilizing Hierholzer's algorithm~\cite{cite_hierholzer_alg} to find the Eulerian circuit. This yields a minimum trail decomposition in $O(|E|)$ time.
\textbf{Y fusion networks:} We use a heuristic algorithm that searches for the longest path within a time limit, removes the path from the graph, and repeats until we have a path cover.
\textbf{XY fusion networks:} We use Algorithm~\ref{alg_find_tc_min_y}.

For our comparative analysis, we employ a simple lower bound for the minimum number of fusions required to implement a graph state. This bound represents the number of fusions required if the fusion network consisted of a single long resource state subdivided as efficiently as possible.

%\comment{Proof should show why we can break it up like this}
\begin{restatable}{lemma}{lemLowerBoundPhotonRestricted}\label{lem_lower_bound_photon_restricted}
Let $G = (E, V)$ be a graph and let $F_{\tx{min}}$ be the minimum number of $XY$ fusions required to implement $G$ with resource states each having at most $L$ photons. Then
\[
F_{\tx{min}} \ge |E| - |V| + \ceil{\frac{2|E| - |V|}{L - 2}}.
\]
\end{restatable}

Figure~\ref{fig_fusions_vs_edges} shows the performance of our algorithms for the different fusion types alongside the lower bound from Lemma~\ref{lem_lower_bound_photon_restricted}, and illustrates two key insights.

First, we are almost always able to find $X$ fusion networks that require fewer fusions than $Y$ fusion networks, with the difference tending to increase with the size of the graph. This is most likely because we are able to compute exact minimum trail decompositions, whereas we can only approximate minimum path covers.
Second, our algorithm for finding $XY$ fusion networks achieves results remarkably close to the lower bound, requiring no more than 8 additional fusions across the QASMBench benchmark set.

The superiority of XY fusions is most evident in the 7 qubit HHL circuit which requires 1466 photons for Y fusions, 1238 photons for X fusions, and only 1142 photons for XY fusions.

\begin{figure}[h]
\centering
\begin{subfigure}[h]{0.8\linewidth}
\includegraphics[width=\linewidth]{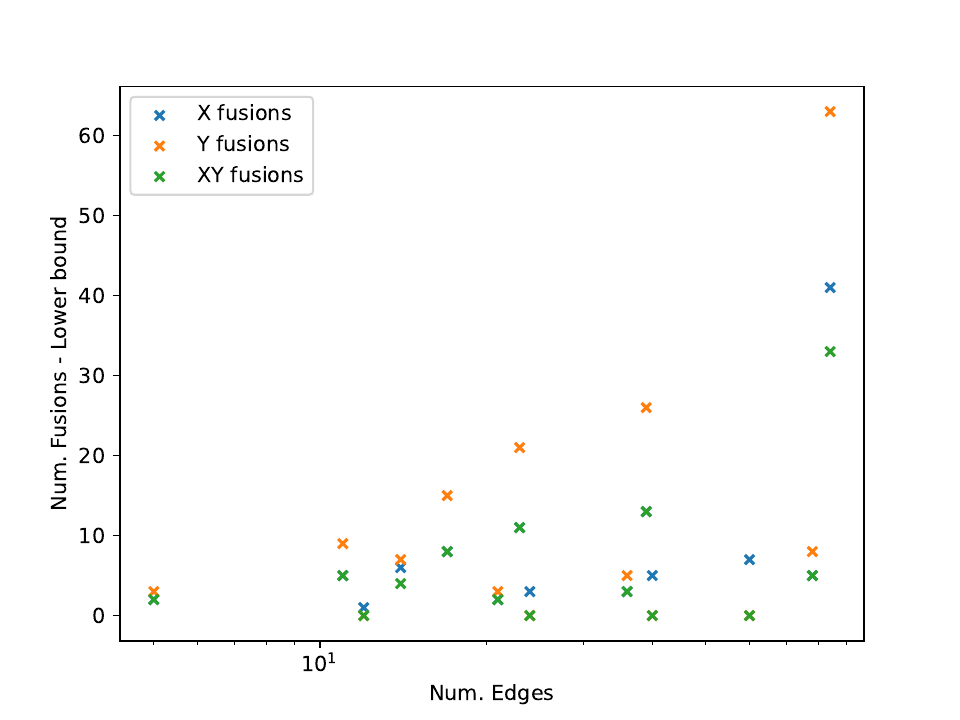}
\caption{B1: QEC}
\end{subfigure}
\begin{subfigure}[h]{0.8\linewidth}
\includegraphics[width=\linewidth]{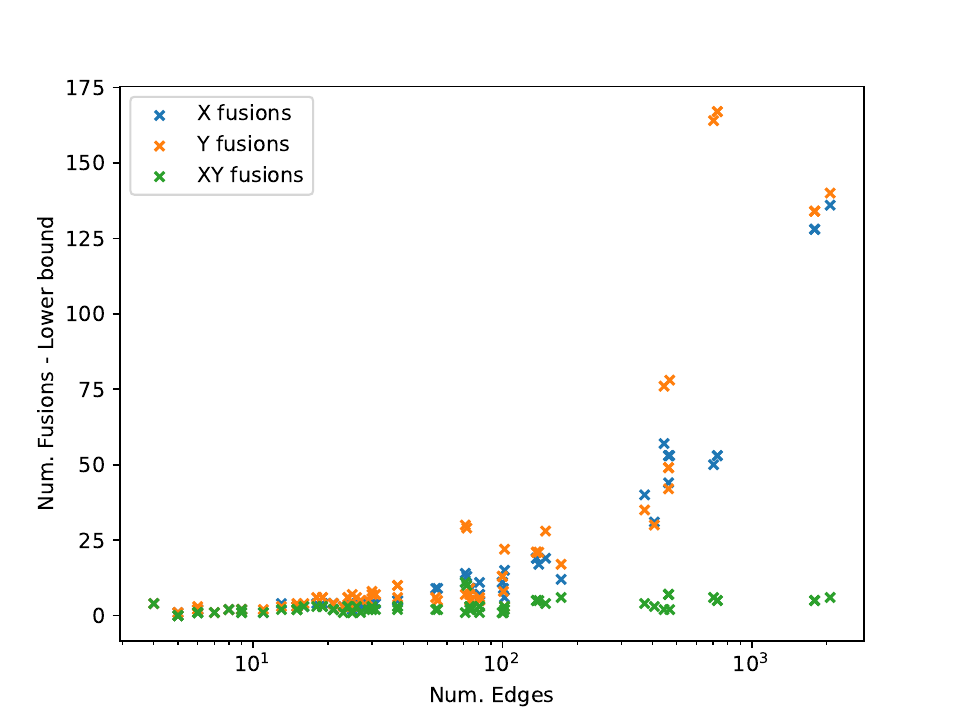}
\caption{B2: QASMBench}
\end{subfigure}
\caption{Difference in number of fusions between fusion networks produced by the heuristic algorithms and the lower bound. As the number of edges increases, the difference grows for $X$ and $Y$ fusion networks, while remaining relatively constant for $XY$ fusion networks.}
\label{fig_fusions_vs_edges}
\end{figure}

Practical implementations often impose restrictions on the number of photons in each linear resource state. Figure~\ref{fig_resource_requirements_shor_encoded_six_ring} illustrates how increasing the number of photons in resource states affects the number of fusions and resource states required. $XY$ fusions demonstrate a clear advantage over either $X$ or $Y$ fusions. However, in this specific case, $Y$ fusions perform significantly better than $X$ fusions because the Shor(2,2) encoded 6-ring graph is cubic, and cubic graphs never require more $Y$ fusions than $X$ fusions when no restrictions are placed on resource state length. In general, the higher the node degree in the graph, the more $Y$ fusions are required relative to $X$ fusions.

% python data_analysis.py plot_ring_shor_encoded_resources_and_photons
\begin{figure}[h]
\centering
\includegraphics[width=0.95\linewidth]{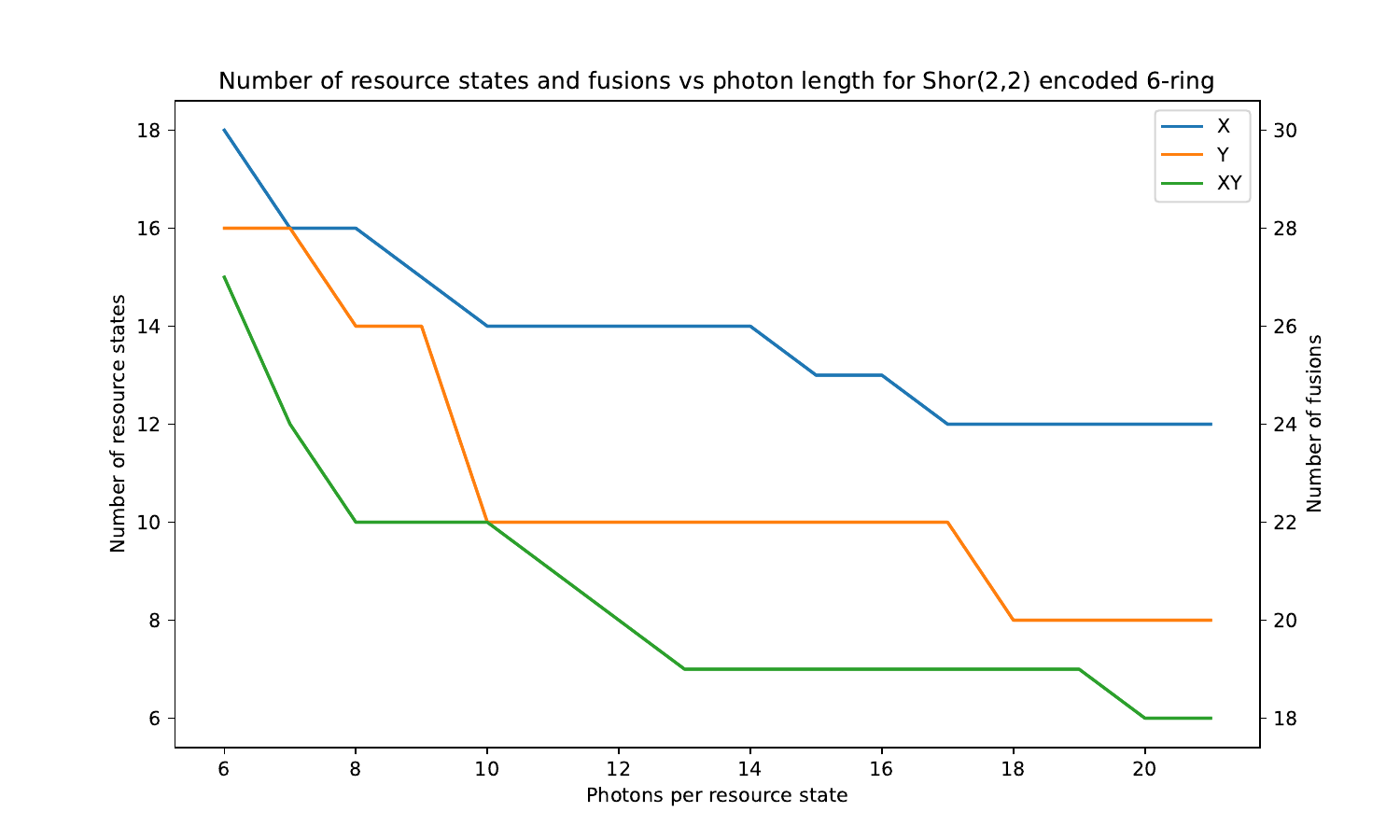}
\caption{Number of fusions and resource states required to implement the Shor(2,2) encoded 6-ring graph state using the algorithms from Section~\ref{sec:approximation} for different resource state lengths. The number of fusions and resource states are linearly dependent from Theorem~\ref{thm_num_fusions}.}
\label{fig_resource_requirements_shor_encoded_six_ring}
\end{figure}

\subsection{Impact of graph rewrites}

% python data_analysis.py
\begin{table*}[t]
    \centering
    \begin{tabular}{|c|ccc|ccc|}
    \hline
    Fusion Type & Fusions & Photons & Resources \\
    \hline
    Y & -10.74\% & -7.16\% & 1.47\% \\
    X & -10.02\% & -6.61\% & -2.32\% \\
    XY & -10.02\% & -6.39\% & 0.75\% \\
    \hline
    \end{tabular}
    \caption{Average percentage change in fusions, photons, and resources required to implement a graph state after applying graph rewrite rules 1, 2, and 3 in the B2 QASMBench benchmark set.}
    \label{tab_avg_graph_rewrites}
\end{table*}

\begin{table*}[t]
\centering
\begin{tabular}{|c|c|c|c|c|}
\hline
Vertices & Before Optim & Greedy & Sim. Annealing & Optimum \\
\hline
3 & 0.50 & 0.00 & 0.00 & 0.00 \\
4 & 1.50 & 0.50 & 0.33 & 0.33 \\
5 & 2.52 & 0.71 & 0.71 & 0.52 \\
6 & 4.14 & 1.67 & 1.50 & 1.05 \\
7 & 5.98 & 2.80 & 2.46 & 1.65 \\
\hline
\end{tabular}
\caption{Average number of X fusions required to implement small graph states of up to $7$ nodes, before and after applying Rewrites 2 and 3. Simulated annealing was performed over 50 iterations.}
\label{tab_benchmarks_l=2}
\end{table*}

Applying the graph rewrite rules from Section~\ref{sec:graphrewrites} further improves performance metrics, as demonstrated in Table~\ref{tab_avg_graph_rewrites}. Our graph rewrite heuristics significantly reduce fusion and photon requirements for implementing fusion networks with unbounded resource states, achieving approximately 11\% reduction in fusions and 9\% reduction in photons. 

The graph rewrites introduced in this work build upon and extend earlier techniques presented in~\cite{cite_qec_benchmarks}, which identified patterns in source graphs where local complementations could reduce the number of resources required in fusion networks consisting of 3-star cluster states (referred to as two-trails in our discussion). Our approach provides a more general framework for applying local complementations that is well-suited to longer resource states.

The number of required resource states remained relatively unchanged in most graphs. A notable exception being the $X$ fusion networks for repeater error correcting codes which increas by more than 50\%. This is because each repeater code contains a fully connected subgraph where every vertex is even. When one of these vertices is locally complemented, it removes almost all edges in the subgraph but makes every vertex odd. Since more edges are removed than odd vertices created, the total number of fusions decreases and the action in accept by our algorithm. However, since the number of resource states equals half the number of odd vertices (Theorem~\ref{thm_min_trail_decomp}), the size of the fusion network is greatly increased.

Repeater graphs could benefit from $n$-ary $Y$ fusions —-- generalized $Y$ fusions applied simultaneously to $n$ nodes with success probability $2^{-(n-1)}$. This approach would improve upon performing $\frac{1}{2}n(n-1)$ pairwise $Y$ fusions or complementing an $n$-clique (success probability $2^{-n}$). Analysis of $n$-ary fusions remains for future work.

Table~\ref{tab_benchmarks_l=2} below reports the number of $X$ fusions required after applying the graph rewrites on all graphs with seven vertices or less, and a variation where we apply the rewrites using a simulated annealing strategy to overcome getting stuck in local minima. Alongside these results we have included the global minimum obtained by exhaustively searching across all possible rewrites. The results highlight the power of graph rewriting, however we note that large improvements can be easily obtained from dense graphs. 

\subsection{Probability of success in RUS architecture}

We now analyze a specific class of architectures comprising caterpillar state generators, routers, repeat-until-success fusion modules~\cite{cite_rus_original}, and measurement modules.

Consider an emitter that deterministically produces caterpillar states with photon emission rate $t_p$ and coherence time $T$. The maximum number of entangled photons it can emit is $L = \floor{T/t_p}$. We assume that any caterpillar state achievable with $L$ photons can be generated.

In addition, we also assume arbitrary photon routing capabilities across different time bins, and fusion measurements with success probability $p_F$ that can be performed through repeat-until-success up to $R$ times. In addition to the fusion failure probability, we incorporate photon loss, assuming every photon has the same probability $\epsilon$ of being lost.

Every time a fusion is attempted, we have the following possible outcomes:
\begin{enumerate}
    \item (success) Both photons are measured and the fusion succeeds, with probability $(1 - p_F)(1 - \epsilon)^2$,
    \item (failure) Both photons are measured and the fusion fails, with probability $p_F(1 - \epsilon)^2$,
    \item (loss) One or zero photons are measured, with probability $\epsilon (2 - \epsilon)$. In this case, the fusion induces Z phase-flip errors on the target qubits with probability $\frac{1}{2}$. 
\end{enumerate}

In a repeat-until-success protocol, we must choose in which of these cases we repeat the computation.

In \cite{wein_minimizing_2024}, the fusion operation is repeated only in the heralded failure case, 
and we terminate the computation every time a photon is lost, or the maximum number of trials is reached.
The RUS fusion success probability then becomes:
\begin{equation}
    p_{RUS}(p_F, \epsilon, R) = (1 - p_F) (1 - \epsilon)^2 \sum_{n = 0}^R (p_F(1 - \epsilon)^2)^n
\end{equation}
This probability post-selects away the loss outcomes of the fusion and thus ensures the successful implementation of the state without the need for error correction.
We call $p_{RUS}$ the \emph{post-selected} probability of success.

In \cite{lee_nearly_2015}, the fusion operation is repeated when either a heralded failure or a loss outcome occurs. 
The computation is only terminated when we reach the maximum number of trials $R$. The probability of success of the RUS fusion is then: 
\begin{equation}\label{eq:prob-upper-bound}
    p^{EC}_{RUS}(p_F, \epsilon, R) = 1 - (1 - p_F(1 - \epsilon)^2)^R
\end{equation}
Note however, that the graph state constructed with this probability will potentially have Z phase-flip errors on its nodes.
The above probability assumes that these phase-flip errors are detected and corrected in an ambient error-correcting code. 
We call $p^{EC}$ the \emph{error corrected} probability of success, it may be seen as an upper bound 
to the success probability of a RUS fusion operation after error correction.
We always have $p \leq p^{EC}$ and $p = p^{EC}$ when $\epsilon = 0$.

The overall success probability for graph state preparation depends on all fusions succeeding. 
With $F$ total fusions, each attempted up to $R$ times, the total success probability becomes:
\[
p_{success} = p_{RUS}(p_F, \epsilon, R)^F 
\]

This objective allows us to formulate the compilation problem as follows:

\begin{definition}[Compiling RUS fusion networks]
Given an MBQC pattern expressed as an open graph $(G, I, O, \lambda, \alpha)$, a maximum number of photons per resource state $L$, fusion success probability $p_F$, and available fusion types ($X$, $Y$, or both), return a fusion network where each resource state contains at most $L$ photons that implements the MBQC pattern with maximum success probability.
\end{definition}

Our overall compilation algorithm proceeds as follows:

\begin{enumerate}
\item Apply graph rewrite rules from Section~\ref{sec:graphrewrites} to reduce the input graph $G$.
\item For each possible number of RUS attempts $R$:
\begin{enumerate}
\item Approximate the minimum photon-restricted trail cover for the reduced graph with at most $L$ photons per resource state.
\item Calculate the success probability $p_{success}$ 
\end{enumerate}

\item Select the value of $R$ that maximizes the overall success probability and return the corresponding fusion network.
\end{enumerate}

% python data_analysis.py plot_spoqc_probabilities 10 
\begin{figure}[h]
\centering
\begin{subfigure}[h]{\linewidth}
\centering
\includegraphics[width=0.8\linewidth]{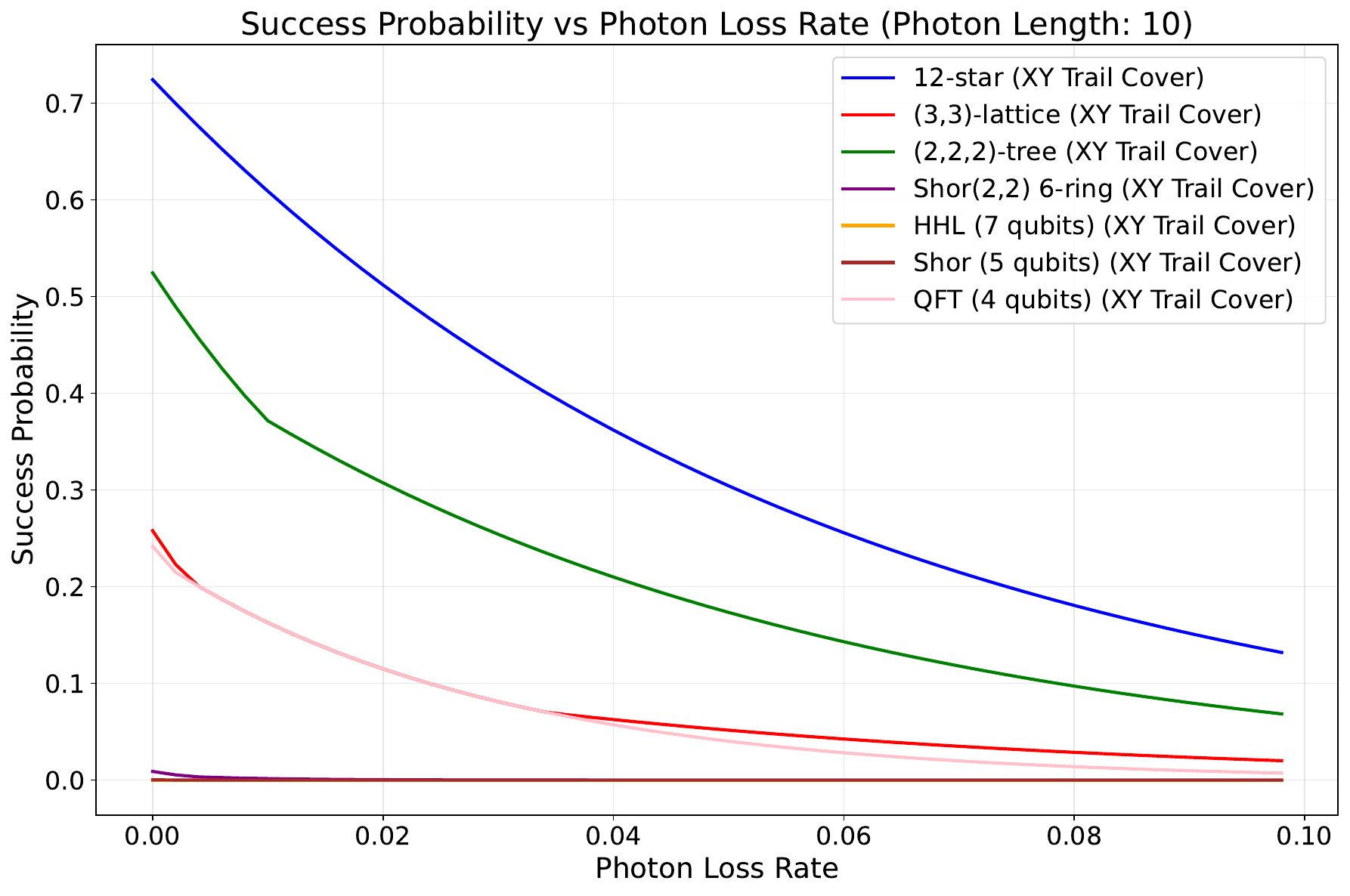}
\end{subfigure}
\begin{subfigure}[h]{\linewidth}
\centering
\includegraphics[width=0.8\linewidth]{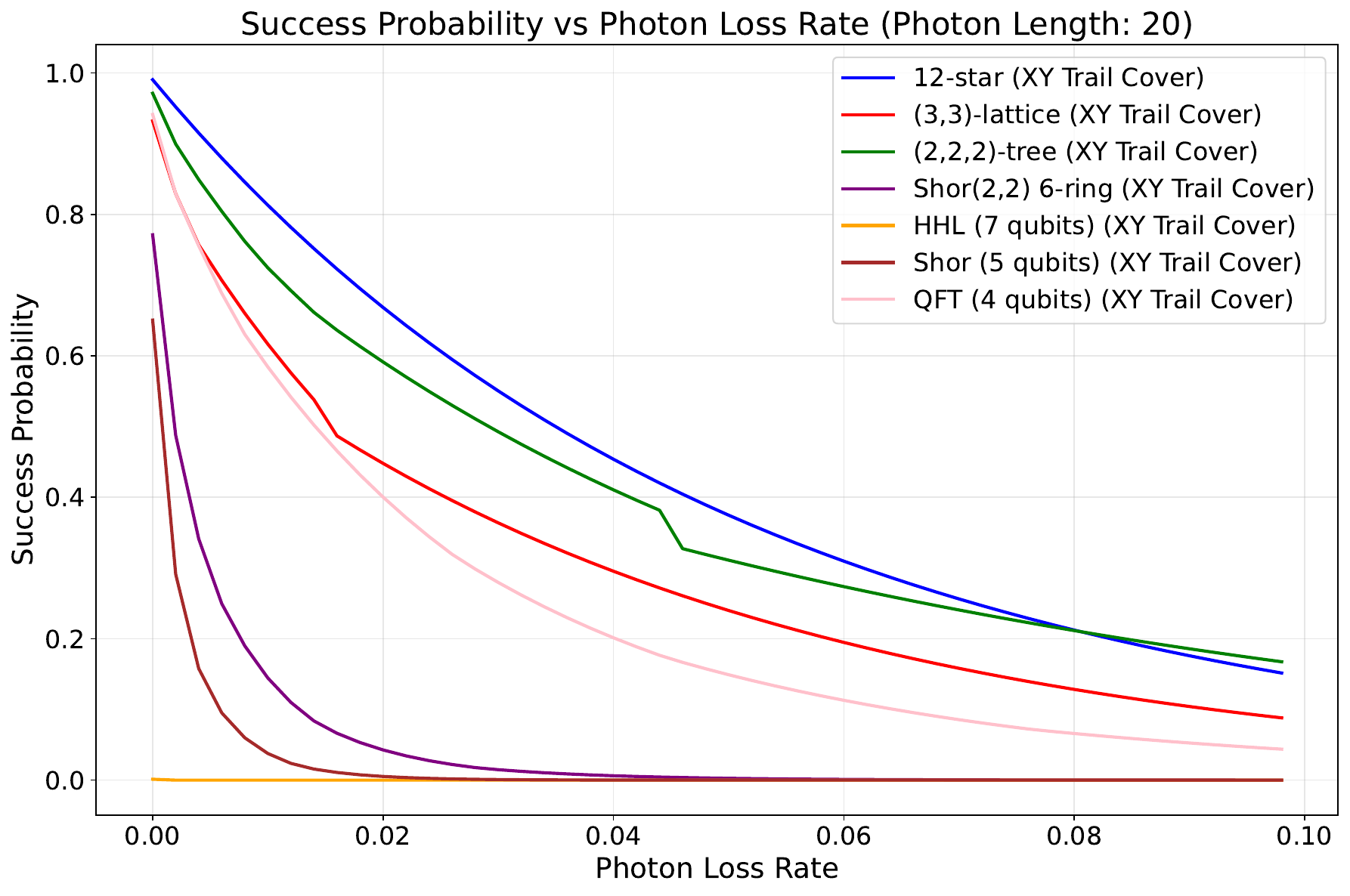}
\end{subfigure}
\caption{Post-selected probability of successfully implementing selected graphs against the photon loss probability for resource states with 10 and 20 photons. The plot for resource states with at most 10 photons appears sparse as most graphs have a near-zero probability of success even without photon loss.}
\label{fig_spoqc_probabilities_vs_photon_loss}
\end{figure}

% python data_analysis.py ??
\begin{figure}[h]
\centering
\includegraphics[width=0.8\linewidth]{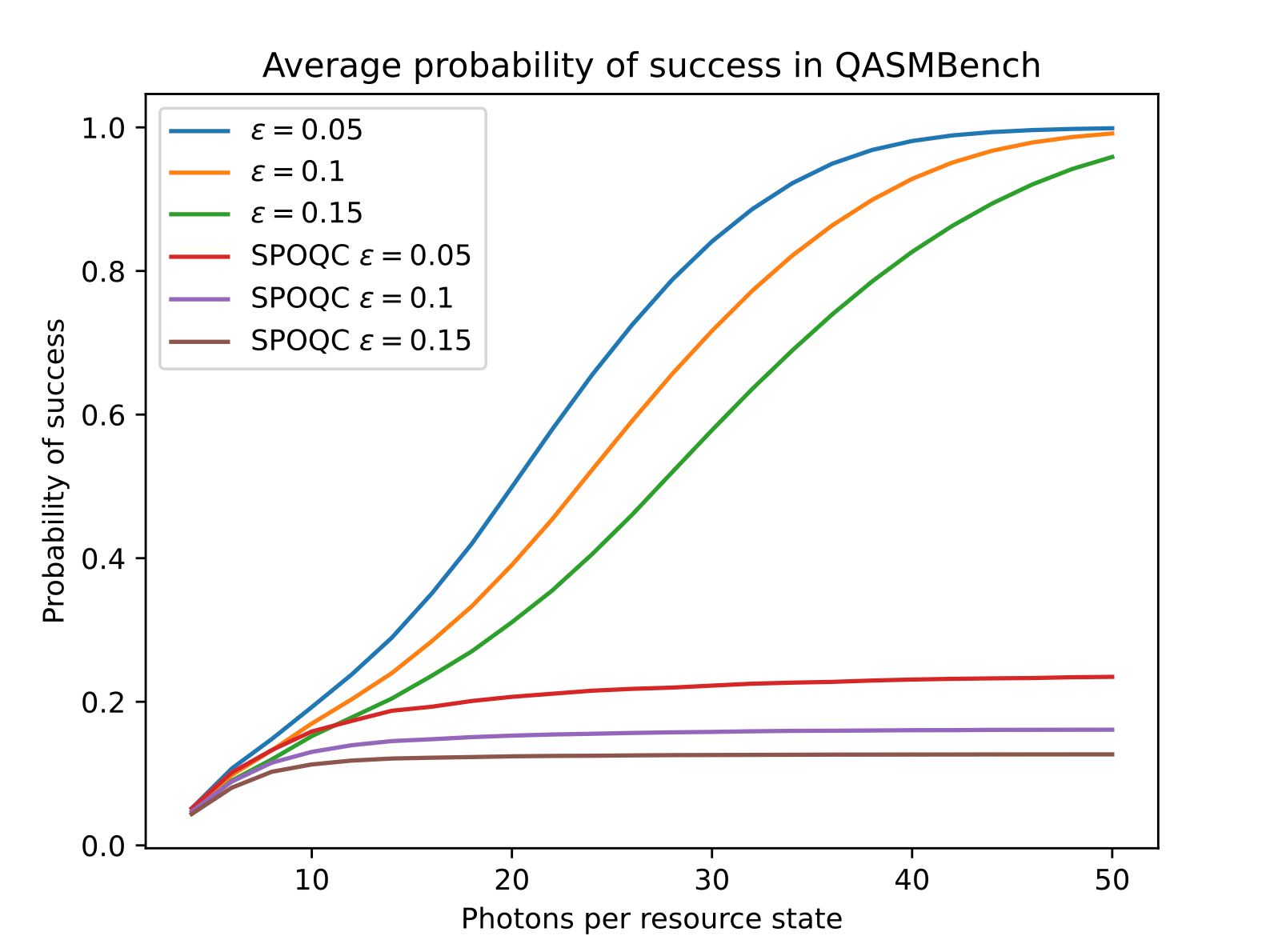}
\caption{Average probability of success for the QASMBench benchmark set comparing post-selected and error-corrected probabilities of success for different photon lengths.}
\label{fig_probabilities_of_success_comparision}
\end{figure}

Figure~\ref{fig_spoqc_probabilities_vs_photon_loss} illustrates the decrease in post-selected probability of success for selected graphs as the probability of photon loss increases.
When resource states are limited to 10 photons most graphs have near-zero post-selected success probability, even without photon loss. 
This changes significantly when we increase the limit to 20 photons per resource state. 

Further increasing the size of resource states beyond 20 photons only marginally iproves the post-selected probability.
This observation is supported by Figure~\ref{fig_probabilities_of_success_comparision}, which compares the average post-selected probability of success with the average error-corrected probability of success. 
The gap between the two plots represents the potential for improvement if errors induced by failed fusions in the repeat-until-success protocol could be corrected.
We observe that the success probability increase steadily up to a critical resource state length after which the benefit of increasing RUS fusions is mitigated by the presence of photon loss.
The post-selected probability of success begins to plateau at around 15 photons per resource state in the QASMBench benchmark set. 
This is unlike the error-corrected model, where increasing repeat-until-success trials plateaus much later, 
showing the potential benefits of larger resource states in a fault-tolerant implementation.

%% file: diagrams/test_graphs_no_zx.tikz
\begin{tikzpicture}
	\begin{pgfonlayer}{nodelayer}
		\node [style=vertex] (0) at (-4.75, 0) {};
		\node [style=vertex] (1) at (-4.25, 0) {};
		\node [style=vertex] (2) at (-3.75, 0) {};
		\node [style=vertex] (4) at (-2.25, 0) {};
		\node [style=vertex] (5) at (-2, 0.5) {};
		\node [style=vertex] (6) at (-1.5, 0.5) {};
		\node [style=vertex] (7) at (-1.25, 0) {};
		\node [style=vertex] (8) at (-1.5, -0.5) {};
		\node [style=vertex] (9) at (-2, -0.5) {};
		\node [style=vertex] (10) at (0.25, 0.5) {};
		\node [style=vertex] (11) at (0.25, 0) {};
		\node [style=vertex] (12) at (0.25, -0.5) {};
		\node [style=vertex] (13) at (0.75, -0.5) {};
		\node [style=vertex] (14) at (1.25, -0.5) {};
		\node [style=vertex] (15) at (1.25, 0) {};
		\node [style=vertex] (16) at (0.75, 0) {};
		\node [style=vertex] (17) at (0.75, 0.5) {};
		\node [style=vertex] (18) at (1.25, 0.5) {};
		\node [style=vertex] (20) at (-4, -2.75) {};
		\node [style=vertex] (21) at (-3.75, -2.25) {};
		\node [style=vertex] (22) at (-3.25, -2.25) {};
		\node [style=vertex] (23) at (-3, -2.75) {};
		\node [style=vertex] (24) at (-3.25, -3.25) {};
		\node [style=vertex] (25) at (-3.75, -3.25) {};
		\node [style=vertex] (26) at (-2.5, -2.75) {};
		\node [style=vertex] (27) at (-3, -1.75) {};
		\node [style=vertex] (28) at (-4, -1.75) {};
		\node [style=vertex] (29) at (-4.5, -2.75) {};
		\node [style=vertex] (30) at (-4, -3.75) {};
		\node [style=vertex] (31) at (-3, -3.75) {};
		\node [style=none] (32) at (-4.25, -1) {4-star};
		\node [style=vertex] (33) at (-4.25, 0.5) {};
		\node [style=vertex] (34) at (-4.25, -0.5) {};
		\node [style=none] (35) at (-3.5, -4.25) {3-repeater};
		\node [style=none] (36) at (-1.75, -1) {6-cycle};
		\node [style=none] (37) at (0.75, -1) {(3,3)-lattice};
		\node [style=vertex] (38) at (-0.25, -3.25) {};
		\node [style=vertex] (39) at (-0.25, -2.75) {};
		\node [style=vertex] (40) at (0.25, -3.25) {};
		\node [style=vertex] (41) at (-0.75, -3.25) {};
		\node [style=vertex] (42) at (0.75, -3) {};
		\node [style=vertex] (43) at (0.75, -3.5) {};
		\node [style=vertex] (44) at (0, -2.25) {};
		\node [style=vertex] (45) at (-0.5, -2.25) {};
		\node [style=vertex] (46) at (-1.25, -3) {};
		\node [style=vertex] (47) at (-1.25, -3.5) {};
		\node [style=vertex] (48) at (1, -2.5) {};
		\node [style=vertex] (49) at (1.25, -3) {};
		\node [style=vertex] (50) at (1.25, -3.5) {};
		\node [style=vertex] (51) at (1, -4) {};
		\node [style=vertex] (52) at (-1.5, -4) {};
		\node [style=vertex] (53) at (-1.75, -3.5) {};
		\node [style=vertex] (54) at (-1.75, -3) {};
		\node [style=vertex] (55) at (-1.5, -2.5) {};
		\node [style=vertex] (56) at (-1, -2) {};
		\node [style=vertex] (57) at (-0.5, -1.75) {};
		\node [style=vertex] (58) at (0, -1.75) {};
		\node [style=vertex] (59) at (0.5, -2) {};
		\node [style=none] (60) at (-0.25, -4.25) {(3,2,2)-tree};
	\end{pgfonlayer}
	\begin{pgfonlayer}{edgelayer}
		\draw (2) to (1);
		\draw (1) to (0);
		\draw (9) to (4);
		\draw (4) to (5);
		\draw (5) to (6);
		\draw (6) to (7);
		\draw (7) to (8);
		\draw (8) to (9);
		\draw (18) to (17);
		\draw (17) to (10);
		\draw (10) to (11);
		\draw (11) to (12);
		\draw (12) to (13);
		\draw (13) to (14);
		\draw (14) to (15);
		\draw (15) to (16);
		\draw (16) to (13);
		\draw (11) to (16);
		\draw (16) to (17);
		\draw (18) to (15);
		\draw (25) to (20);
		\draw (20) to (21);
		\draw (21) to (22);
		\draw (22) to (23);
		\draw (23) to (24);
		\draw (24) to (25);
		\draw (21) to (24);
		\draw (25) to (22);
		\draw (23) to (20);
		\draw (31) to (24);
		\draw (25) to (30);
		\draw (20) to (29);
		\draw (28) to (21);
		\draw (22) to (27);
		\draw (26) to (23);
		\draw (34) to (1);
		\draw (1) to (33);
		\draw (38) to (40);
		\draw (40) to (42);
		\draw (42) to (48);
		\draw (42) to (49);
		\draw (50) to (43);
		\draw (43) to (40);
		\draw (51) to (43);
		\draw (38) to (41);
		\draw (41) to (47);
		\draw (47) to (52);
		\draw (53) to (47);
		\draw (41) to (46);
		\draw (46) to (54);
		\draw (46) to (55);
		\draw (38) to (39);
		\draw (39) to (44);
		\draw (44) to (59);
		\draw (44) to (58);
		\draw (39) to (45);
		\draw (45) to (57);
		\draw (45) to (56);
	\end{pgfonlayer}
\end{tikzpicture}

%% file: appendix.tex
\section{Bounded Trail cover estimations}

\begin{restatable}{lemma}{lemNumberTheory}\label{lem_number_theory}
Let $(t_i)_{i=1}^N$ and $L$ be positive integers. Then
\begin{equation}\label{eq_mod_in_ceiling}
\sum_{i = 1}^N \ceil{\frac{t_i}{L}} - \ceil{\sum_{i = 1}^N \frac{t_i}{L}} = \sum_{i = 1}^N \ceil{\frac{t_i \Mod L}{L}} - \ceil{\sum_{i = 1}^N \frac{t_i \Mod L}{L}}.
\end{equation}
\end{restatable}

\begin{proof}
For positive integers $a$ and $b$, we have $\frac{a}{b} = \floor{\frac{a}{b}} + \frac{a \Mod b}{b}$. Therefore
\begin{equation}\label{eq_lem_number_theory_1}
\sum_{i = 1}^N \ceil{\frac{t_i}{L}} = \sum_{i = 1}^N \floor{\frac{t_i}{L}} + \sum_{i = 1}^N \ceil{\frac{t_i \Mod L}{L}}.
\end{equation}
and
\begin{equation}\label{eq_lem_number_theory_2}
\ceil{\sum_{i = 1}^N \frac{t_i}{L}} = \sum_{i = 1}^N \floor{\frac{t_i}{L}} + \ceil{\sum_{i = 1}^N \frac{t_i\Mod L}{L}}
\end{equation}
Subtracting~\eqref{eq_lem_number_theory_2} from~\eqref{eq_lem_number_theory_1} gives
\begin{align*}
\sum_{i = 1}^N \ceil{\frac{t_i}{L}} - \ceil{\sum_{i = 1}^N \frac{t_i}{L}} &= \sum_{i = 1}^N \ceil{\frac{t_i \Mod L}{L}} - \ceil{\sum_{i = 1}^N \frac{t_i \Mod L}{L}}.
\end{align*}
\end{proof}

\lemNumberTheoryBounds*

\begin{proof}
By Lemma~\ref{lem_number_theory}, it is sufficient to show that
\begin{equation}\label{eq_optimise_ceilings}
\sum_{i = 1}^N \ceil{\frac{t_i \Mod L}{L}} - \ceil{\sum_{i = 1}^N \frac{t_i \Mod L}{L}}
\end{equation}
maximised when $t_i \Mod L = 1$ for all $1 \le i \le N$ and is equal to $N - \ceil{\frac{N}{L}}$.

Suppose that for some $1 \le i \le N$ we have $t_i \Mod L = 0$. Then if instead we had $t_i \Mod L = 1$, the first sum of~\eqref{eq_optimise_ceilings} increases by one and the second sum increase by at most one. Hence~\eqref{eq_optimise_ceilings} does not decrease.
Suppose now that $t_i \Mod L > 1$. Then if $t_i \Mod L = 1$, first sum in~\eqref{eq_optimise_ceilings} will not change and the second sum may decrease by one. Hence~\eqref{eq_optimise_ceilings} will not decrease in this case either.

Therefore the configuration where $t_1 \Mod L = \cdots = t_N \Mod L = 1$ is no less than any other assignment of $(t_i)_{i=1}^N$ and is therefore the maximum. Substituting these values into~\eqref{eq_optimise_ceilings} gives us the desired equality.
\end{proof}

\propAverageLTrailApprox*

\begin{proof}
Let $\M{T}$ be a minimum trail decomposition, and suppose that suppose that $N$ of the trails in $\M{T}$ are a multiple of $L$. Then on average the remaining trails $T_i$ have length $T_i \mod L = \frac{1}{2}L$. Therefore substituting into~\eqref{eq_optimise_ceilings} we have at most
\[
K -N - \frac{(K - N)(\frac{L}{2})}{L} = \frac{K-N}{2}
\]
more trails than the minimum. This reaches a maximum of $\frac{K}{2}$ when $N = 0$. Therefore the expected number of trails is at most $\frac{K}{2}$ more than the minimum. Substituting $K = \frac{1}{2}|\Odd(G)|$ gives the result.
\end{proof}

\propConvertToMTD*
\begin{proof}
Let $\M{T}$ be a trail decomposition of $G$ obtained by applying the two rules on some trail decomposition. We will show that all trails in the decomposition satisfy the definition of belonging to a minimum trail decomposition.

Let $T \in \M{T}$. If $T$ is the only trail in its connected component, then it belongs to a minimum trail decomposition and we are done. Now suppose that $T$ is not the only trail in its connected component. Then no other trail in $\M{T}$ can end at the same vertices that $T$ ends at since otherwise we could have joined them. This also implies that $T$ ends at distinct odd vertices or both endpoints are at an even vertex and $T$ is closed. However, if $T$ is closed we would have been able to apply the second rule to join it with another trail in the connected component. Thus $T$ ends at distinct odd vertices. Since $T$ was arbitrary, this implies that every trail in the connected component of $T$ also ends at distinct odd vertices and therefore $\M{T}$ is a minimum trail decomposition on the connected component of $T$. We can therefore conclude that every $\M{C}T$ is minimum trail decomposition on any connected component of $G$ and thus $\M{T}$ is a minimum trail decomposition on $G$.
\end{proof}

\propLTrailsAreSubdivisions*

\begin{proof}
Suppose we have a minimum $L$-trail decomposition $\M{T}$ of $G$. Then by using the reduction in Proposition~\ref{prop_convert_to_mtd}, we are able to convert it into a minimum trail decomposition $\M{T}^\prime$ of $G$.
Since for each trail $T$ in $\M{T}$, every edge in $T$ is included in exactly one trail in $\M{T}^\prime$, subdividing each trail of $\M{T}^\prime$ will produce a minimum $L$-trail decomposition.
\end{proof}

\propLTrailApprox*

\begin{proof}
Suppose we have a minimum trail decomposition $\M{T} = \{T_1, \ldots, T_K\}$ for some integer $K$. Then subdividing each trail into $L$-trails produces an $L$-trail decomposition of size $\sum^K_{i=1} \ceil{\frac{|T_i|}{L}}$ where $|T_i|$ is the number of edges in the trail $T_i$.

A lower bound for the minimum number of trails in an $L$-trail decomposition is $\ceil{\frac{|E|}{L}}$, or equivalently $\ceil{\frac{1}{L}\sum_{i = 1}^K |T_i|}$. Hence the difference between the number of trails in $\M{T}$ and in the minimum $L$-trail decomposition is at most
\[
\sum_{i = 1}^K \ceil{\frac{|T_i|}{L}} - \ceil{\sum_{i = 1}^K \frac{|T_i|}{L}} \le K - \ceil{\frac{K}{L}} = \floor{K\left(1 - \frac{1}{L}\right)}.
\]

where the inequality follows from an application of Lemma~\ref{lem_number_theory_bounds}. 

From Theorem~\ref{thm_min_trail_decomp} we know that $K = \frac{1}{2}|\Odd(G)|$ if $|\Odd(G)| > 0$ and $K = 1$ otherwise. However, in the case where $|\Odd(G)| = 0$, this subdivision gives a minimum $L$-trail decomposition anyway, so the proposition still holds.
\end{proof}

\propPhotonTrailApprox*

\begin{proof}
Follow a similar argument to the proof of Proposition~\ref{prop_ltrail_approx}.

Suppose we have a minimum trail decomposition $\M{T} = \{T_1, \ldots, T_K\}$ that implements the open graph $G$, and suppose the resource state corresponding to the trail $T_i$ contain $P_i$ photons.
Then if we subdivide each resource state so that each has $L$ photons, we will have
\[
\sum^K_{i=1} \ceil{\frac{P_i - 2}{L - 2}}
\]
resource states in total where $K = \frac{1}{2}|\Odd(G)|$ if $|\Odd(G)| > 0$ and $K = 1$ otherwise. A lower bound for the minimum number of resource states is $\ceil{\frac{(\sum^K_{i=1} P_i) - 2}{L-2}}$. Therefore the difference between the number of resource states in $\M{T}$ and in the minimum trail decomposition is at most

\begin{align*}
\sum^K_{i=1} \ceil{\frac{P_i - 2}{L - 2}} &- \ceil{\frac{(\sum^K_{i=1} P_i) - 2}{L - 2}} \\
&= \sum^K_{i=1} \ceil{\frac{P_i - 2}{L - 2}} - \ceil{\frac{\sum^K_{i=1} P_i - 2}{L - 2} + \frac{2K - 2}{L - 2}} \\
&\le \sum^K_{i=1} \ceil{\frac{P_i - 2}{L - 2}} - \ceil{\frac{\sum^K_{i=1} P_i - 2}{L - 2}} - \ceil{\frac{2K - 2}{L - 2}} \\
&\le K - \ceil{\frac{K}{L-2}} - \ceil{\frac{2K - 2}{L - 2}} \\
&\le K - \ceil{\frac{3K - 2}{L-2}} + 1 \\
&\le K - \frac{3K - 2}{L-2} + 1 \\
&\le K - \frac{3K}{L-2} + 1
\end{align*}
Substituting $K = \frac{1}{2}|\Odd(G)|$ completes the proof.
\end{proof}

\section{Compilation}

%TODO haven't shown how to subdivide in the photon bounded case.
\begin{lemma}\label{lem_subdivision_eq}
Subdividing a linear resource state with $P \ge 3$ photons into resource states each with at most $L \ge 3$ photons using either $X$ or $Y$ fusions will result in
\begin{equation}\label{eq_photon_subdivision}
\ceil{\frac{P - 2}{L-2}}
\end{equation}
resource states.
\end{lemma}

\begin{proof}
If $P \le L$ then \eqref{eq_photon_subdivision} holds. Now assume $P > L$, then after subdivision the two ends of the resource state will belong to distinct resource states, each with $L-1$ photons from the original resource state and one additional photon for the fusion arising from the subdivision. The remaining part of the original resource state without the two ends will have $P - 2(L - 1)$ photons and be subdivided into pieces containing $L - 2$ photons from the original resource state and two additional photons, one on either end for fusions arising from the subdivision. This will therefore be split into
\[
\ceil{\frac{P - 2(L - 1)}{L-2}} = \ceil{\frac{P - 2}{L-2}} - 2
\]
resource states. Adding the two resource states on either end gives us \eqref{eq_photon_subdivision}.
\end{proof}

\lemLowerBoundPhotonRestricted*

\begin{proof}
In the resources state, there is one photon for each node in the final graph for measurement or output and two photons for every fusion. Therefore there are
\[
P = 2F + |V| = 2|E| - |V| + 2.
\]
in the original resource state.
After subdivision, we may apply Lemma~\ref{lem_subdivision_eq} to the fusion equation to compute the number of fusions as
\[
F = |E| - |V| + \ceil{\frac{P - 2}{L - 2}} = |E| - |V| + \ceil{\frac{2|E| - |V|}{L - 2}}
\]
fusions.
\end{proof}

\section{Counterexample for triangle complementation for bounded trail decompositions}\label{appendix_sec_tri_complementation_counterexample}

Complementing the triangle does not always produce a graph that admits a smaller bounded trail decomposition. The counterexample below illustrates a minimum $4$-trail decomposition on a graph.

\begin{figure}[H]
\centering
\input{diagrams/comp_triangle_counterexample.tikz}
\end{figure}

If we complement the central triangle, we get the following graph with its minimum 4-trail decomposition.

\begin{figure}[H]
\centering
\input{diagrams/comp_triangle_counterexample_complemented.tikz}
\end{figure}

The original graph required six 4-trails whereas after complementing the triangle, the new graph required eight.

%% file: diagrams/comp_triangle_counterexample.tikz
\begin{tikzpicture}
	\begin{pgfonlayer}{nodelayer}
		\node [style=vertex] (0) at (0, 2) {};
		\node [style=vertex] (1) at (0.75, 2) {};
		\node [style=vertex] (2) at (-0.25, 2.5) {};
		\node [style=vertex] (3) at (-0.5, 3) {};
		\node [style=vertex] (4) at (0.25, 2.5) {};
		\node [style=vertex] (5) at (0.5, 3) {};
		\node [style=vertex] (6) at (0.75, 2.5) {};
		\node [style=vertex] (7) at (-0.5, 1) {};
		\node [style=vertex] (8) at (0.5, 1) {};
		\node [style=vertex] (9) at (1, 1) {};
		\node [style=vertex] (10) at (1.5, 1) {};
		\node [style=vertex] (11) at (1, 0.5) {};
		\node [style=vertex] (12) at (1, 0) {};
		\node [style=vertex] (13) at (1.5, 0.25) {};
		\node [style=vertex] (14) at (0.5, 0.5) {};
		\node [style=vertex] (15) at (-0.5, 0.5) {};
		\node [style=vertex] (17) at (-0.5, 0) {};
		\node [style=vertex] (18) at (-1, 1) {};
		\node [style=vertex] (19) at (-1, 0) {};
		\node [style=vertex] (20) at (-1.5, 0.25) {};
		\node [style=vertex] (21) at (-1, 0.5) {};
		\node [style=none] (76) at (2.25, 1.5) {$\rightsquigarrow$};
		\node [style=vertex] (81) at (3.25, 1.25) {};
		\node [style=vertex] (82) at (3.75, 1.25) {};
		\node [style=vertex] (83) at (4.25, 2) {};
		\node [style=vertex] (84) at (4, 2.5) {};
		\node [style=vertex] (85) at (3.75, 3) {};
		\node [style=vertex] (86) at (4.75, 2.25) {};
		\node [style=vertex] (87) at (4.75, 3.25) {};
		\node [style=vertex] (88) at (4.75, 2.75) {};
		\node [style=vertex] (89) at (5.25, 2) {};
		\node [style=vertex] (90) at (5.75, 2) {};
		\node [style=vertex] (91) at (5.75, 1.25) {};
		\node [style=vertex] (92) at (6.25, 1.25) {};
		\node [style=vertex] (93) at (6.75, 1.25) {};
		\node [style=vertex] (94) at (5.75, 0.75) {};
		\node [style=vertex] (95) at (6.25, 0.5) {};
		\node [style=vertex] (96) at (6.25, 0) {};
		\node [style=vertex] (97) at (5.25, 0) {};
		\node [style=vertex] (98) at (5.25, 0.5) {};
		\node [style=vertex] (99) at (4.25, 0.5) {};
		\node [style=vertex] (100) at (4.25, 0) {};
		\node [style=vertex] (101) at (4.25, -0.5) {};
		\node [style=vertex] (102) at (3.75, 0.75) {};
		\node [style=vertex] (103) at (3.25, 0.5) {};
		\node [style=vertex] (104) at (3.25, 0) {};
		\node [style=vertex] (105) at (5.25, 3.25) {};
		\node [style=vertex] (106) at (5.25, 2.75) {};
		\node [style=vertex] (107) at (2.75, 0) {};
		\node [style=vertex] (108) at (2.75, 0.5) {};
		\node [style=vertex] (109) at (6.75, 0) {};
		\node [style=vertex] (110) at (6.75, 0.5) {};
		\node [style=none] (111) at (0, -1.25) {Original};
		\node [style=none] (112) at (5, -1.25) {Minimum 4-trail decomposition};
	\end{pgfonlayer}
	\begin{pgfonlayer}{edgelayer}
		\draw (8) to (9);
		\draw (9) to (10);
		\draw (8) to (14);
		\draw (8) to (11);
		\draw (11) to (13);
		\draw (13) to (12);
		\draw (12) to (11);
		\draw (7) to (8);
		\draw (8) to (0);
		\draw (0) to (2);
		\draw (2) to (3);
		\draw (0) to (1);
		\draw (5) to (6);
		\draw (6) to (4);
		\draw (4) to (5);
		\draw (4) to (0);
		\draw (0) to (7);
		\draw (19) to (20);
		\draw (21) to (19);
		\draw (21) to (20);
		\draw (21) to (7);
		\draw (7) to (15);
		\draw (15) to (17);
		\draw (7) to (18);
		\draw (88) to (87);
		\draw (88) to (86);
		\draw (85) to (84);
		\draw (84) to (83);
		\draw (83) to (82);
		\draw (82) to (81);
		\draw (90) to (89);
		\draw (89) to (91);
		\draw (91) to (92);
		\draw (92) to (93);
		\draw (96) to (95);
		\draw (95) to (94);
		\draw (101) to (100);
		\draw (100) to (99);
		\draw (99) to (98);
		\draw (98) to (97);
		\draw (104) to (103);
		\draw (103) to (102);
		\draw [style=X fusion] (83) to (86);
		\draw [style=X fusion] (86) to (89);
		\draw [style=X fusion] (91) to (94);
		\draw [style=X fusion] (94) to (98);
		\draw [style=X fusion] (99) to (102);
		\draw [style=X fusion] (102) to (82);
		\draw [style=X fusion] (106) to (88);
		\draw (105) to (106);
		\draw (105) to (87);
		\draw (108) to (107);
		\draw (107) to (104);
		\draw [style=X fusion] (108) to (103);
		\draw [style=X fusion] (110) to (95);
		\draw (110) to (109);
		\draw (109) to (96);
	\end{pgfonlayer}
\end{tikzpicture}

%% file: diagrams/comp_triangle_counterexample_complemented.tikz
\begin{tikzpicture}
	\begin{pgfonlayer}{nodelayer}
		\node [style=vertex] (0) at (2, 2) {};
		\node [style=vertex] (1) at (2.75, 2) {};
		\node [style=vertex] (2) at (1.75, 2.5) {};
		\node [style=vertex] (3) at (1.5, 3) {};
		\node [style=vertex] (4) at (2.25, 2.5) {};
		\node [style=vertex] (5) at (2.5, 3) {};
		\node [style=vertex] (6) at (2.75, 2.5) {};
		\node [style=vertex] (7) at (1.5, 1) {};
		\node [style=vertex] (8) at (2.5, 1) {};
		\node [style=vertex] (9) at (3, 1) {};
		\node [style=vertex] (10) at (3.5, 1) {};
		\node [style=vertex] (11) at (3, 0.5) {};
		\node [style=vertex] (12) at (3, 0) {};
		\node [style=vertex] (13) at (3.5, 0.25) {};
		\node [style=vertex] (14) at (2.5, 0.5) {};
		\node [style=vertex] (15) at (1.5, 0.5) {};
		\node [style=vertex] (16) at (1.5, 0) {};
		\node [style=vertex] (17) at (1, 1) {};
		\node [style=vertex] (18) at (1, 0) {};
		\node [style=vertex] (19) at (0.5, 0.25) {};
		\node [style=vertex] (20) at (1, 0.5) {};
		\node [style=none] (51) at (4.25, 1.5) {$\rightsquigarrow$};
		\node [style=vertex] (52) at (2, 1.5) {};
		\node [style=vertex] (53) at (5.5, 1.25) {};
		\node [style=vertex] (54) at (6, 1.25) {};
		\node [style=vertex] (55) at (6.5, 2) {};
		\node [style=vertex] (56) at (6.5, 2.5) {};
		\node [style=vertex] (57) at (7, 2) {};
		\node [style=vertex] (58) at (7, 3) {};
		\node [style=vertex] (59) at (7, 2.5) {};
		\node [style=vertex] (60) at (7.5, 2) {};
		\node [style=vertex] (61) at (8, 2) {};
		\node [style=vertex] (62) at (8.5, 0.75) {};
		\node [style=vertex] (63) at (9, 0.75) {};
		\node [style=vertex] (64) at (9.5, 0.75) {};
		\node [style=vertex] (65) at (8, 0.75) {};
		\node [style=vertex] (66) at (8, 0.25) {};
		\node [style=vertex] (67) at (7.5, 0.25) {};
		\node [style=vertex] (68) at (7.5, 1) {};
		\node [style=vertex] (69) at (6.5, 0.5) {};
		\node [style=vertex] (70) at (6.5, 0) {};
		\node [style=vertex] (71) at (6.5, -0.5) {};
		\node [style=vertex] (72) at (6, 0.75) {};
		\node [style=vertex] (73) at (5.5, 0.5) {};
		\node [style=vertex] (74) at (5.5, 0) {};
		\node [style=vertex] (75) at (7.5, 3) {};
		\node [style=vertex] (76) at (7.5, 2.5) {};
		\node [style=vertex] (77) at (5, 0) {};
		\node [style=vertex] (78) at (5, 0.5) {};
		\node [style=vertex] (79) at (7.5, -0.25) {};
		\node [style=vertex] (80) at (8, -0.25) {};
		\node [style=vertex] (81) at (8.5, 2) {};
		\node [style=vertex] (82) at (6.5, 1.5) {};
		\node [style=vertex] (83) at (7, 1.25) {};
		\node [style=vertex] (84) at (8.5, 0.25) {};
		\node [style=none] (85) at (2, -1.25) {Original};
		\node [style=none] (86) at (7, -1.25) {Minimum 4-trail decomposition};
	\end{pgfonlayer}
	\begin{pgfonlayer}{edgelayer}
		\draw (8) to (9);
		\draw (9) to (10);
		\draw (8) to (14);
		\draw (8) to (11);
		\draw (11) to (13);
		\draw (13) to (12);
		\draw (12) to (11);
		\draw (0) to (2);
		\draw (2) to (3);
		\draw (0) to (1);
		\draw (5) to (6);
		\draw (6) to (4);
		\draw (4) to (5);
		\draw (4) to (0);
		\draw (18) to (19);
		\draw (20) to (18);
		\draw (20) to (19);
		\draw (20) to (7);
		\draw (7) to (15);
		\draw (15) to (16);
		\draw (7) to (17);
		\draw (52) to (8);
		\draw (52) to (7);
		\draw (52) to (0);
		\draw (59) to (58);
		\draw (59) to (57);
		\draw (56) to (55);
		\draw (54) to (53);
		\draw (61) to (60);
		\draw (62) to (63);
		\draw (63) to (64);
		\draw (67) to (66);
		\draw (66) to (65);
		\draw (71) to (70);
		\draw (70) to (69);
		\draw (74) to (73);
		\draw (73) to (72);
		\draw [style=X fusion] (55) to (57);
		\draw [style=X fusion] (57) to (60);
		\draw [style=X fusion] (62) to (65);
		\draw [style=X fusion] (65) to (68);
		\draw [style=X fusion] (69) to (72);
		\draw [style=X fusion] (72) to (54);
		\draw [style=X fusion] (76) to (59);
		\draw (75) to (76);
		\draw (75) to (58);
		\draw (78) to (77);
		\draw (77) to (74);
		\draw [style=X fusion] (78) to (73);
		\draw [style=X fusion] (80) to (66);
		\draw (80) to (79);
		\draw (79) to (67);
		\draw (81) to (61);
		\draw (82) to (55);
		\draw (82) to (54);
		\draw (83) to (68);
		\draw [style=X fusion] (82) to (83);
		\draw (84) to (62);
	\end{pgfonlayer}
\end{tikzpicture}

%% file: main.bbl
\begin{thebibliography}{10}

\bibitem{raussendorf_one-way_2001}
Robert Raussendorf and Hans~J. Briegel.
\newblock ``A {One}-{Way} {Quantum} {Computer}''.
\newblock \href{https://dx.doi.org/10.1103/PhysRevLett.86.5188}{Physical Review Letters {\bf 86}, 5188--5191}~(2001).

\bibitem{stanisic_generating_2017}
Stasja Stanisic, Noah Linden, Ashley Montanaro, and Peter~S. Turner.
\newblock ``Generating {Entanglement} with {Linear} {Optics}''.
\newblock \href{https://dx.doi.org/10.1103/PhysRevA.96.043861}{Physical Review A {\bf 96}, 043861}~(2017).

\bibitem{Browne_2005}
Daniel~E. Browne and Terry Rudolph.
\newblock ``Resource-efficient linear optical quantum computation''.
\newblock \href{https://dx.doi.org/10.1103/physrevlett.95.010501}{Physical Review Letters{\bf 95}}~(2005).

\bibitem{repub_fbqc}
Sara Bartolucci, Patrick Birchall, Hector Bomb{\'\i}n, Hugo Cable, Chris Dawson, Mercedes Gimeno-Segovia, Eric Johnston, Konrad Kieling, Naomi Nickerson, Mihir Pant, Fernando Pastawski, Terry Rudolph, and Chris Sparrow.
\newblock ``Fusion-based quantum computation''.
\newblock \href{https://dx.doi.org/10.1038/s41467-023-36493-1}{Nature Communications {\bf 14}, 912}~(2023).

\bibitem{cite_ft_quantum_computing_with_nondeterminism}
Ying Li, Sean~D. Barrett, Thomas~M. Stace, and Simon~C. Benjamin.
\newblock ``Fault tolerant quantum computation with nondeterministic gates''.
\newblock \href{https://dx.doi.org/10.1103/PhysRevLett.105.250502}{Phys. Rev. Lett. {\bf 105}, 250502}~(2010).

\bibitem{cite_photonic_architecture_with_ghz}
Srikrishna Omkar, Seok-Hyung Lee, Yong~Siah Teo, Seung-Woo Lee, and Hyunseok Jeong.
\newblock ``All-photonic architecture for scalable quantum computing with greenberger-horne-zeilinger states''.
\newblock \href{https://dx.doi.org/10.1103/PRXQuantum.3.030309}{PRX Quantum {\bf 3}, 030309}~(2022).

\bibitem{de_jong_extracting_2024}
J.~de~Jong, F.~Hahn, N.~Tcholtchev, M.~Hauswirth, and A.~Pappa.
\newblock ``Extracting {GHZ} states from linear cluster states''.
\newblock \href{https://dx.doi.org/10.1103/PhysRevResearch.6.013330}{Physical Review Research {\bf 6}, 013330}~(2024).

\bibitem{blinov_observation_2004}
B.~B. Blinov, D.~L. Moehring, L.-M. Duan, and C.~Monroe.
\newblock ``Observation of entanglement between a single trapped atom and a single photon''.
\newblock \href{https://dx.doi.org/10.1038/nature02377}{Nature {\bf 428}, 153--157}~(2004).

\bibitem{istrati_sequential_2020}
D.~Istrati, Y.~Pilnyak, J.~C. Loredo, C.~Antón, N.~Somaschi, P.~Hilaire, H.~Ollivier, M.~Esmann, L.~Cohen, L.~Vidro, C.~Millet, A.~Lemaître, I.~Sagnes, A.~Harouri, L.~Lanco, P.~Senellart, and H.~S. Eisenberg.
\newblock ``Sequential generation of linear cluster states from a single photon emitter''.
\newblock \href{https://dx.doi.org/10.1038/s41467-020-19341-4}{Nature Communications {\bf 11}, 5501}~(2020).

\bibitem{quant-dots}
N.~Coste, D.~A. Fioretto, N.~Belabas, S.~C. Wein, P.~Hilaire, R.~Frantzeskakis, M.~Gundin, B.~Goes, N.~Somaschi, M.~Morassi, A.~Lema{\^\i}tre, I.~Sagnes, A.~Harouri, S.~E. Economou, A.~Auffeves, O.~Krebs, L.~Lanco, and P.~Senellart.
\newblock ``High-rate entanglement between a semiconductor spin and indistinguishable photons''.
\newblock \href{https://dx.doi.org/10.1038/s41566-023-01186-0}{Nature Photonics {\bf 17}, 582--587}~(2023).

\bibitem{bartolucci_creation_2021}
Sara Bartolucci, Patrick~M. Birchall, Mercedes Gimeno-Segovia, Eric Johnston, Konrad Kieling, Mihir Pant, Terry Rudolph, Jake Smith, Chris Sparrow, and Mihai~D. Vidrighin.
\newblock ``Creation of {Entangled} {Photonic} {States} {Using} {Linear} {Optics}''~(2021).
\newblock arXiv:2106.13825.

\bibitem{og_fbqc}
Sara Bartolucci, Patrick Birchall, Hector Bombin, Hugo Cable, Chris Dawson, Mercedes Gimeno-Segovia, Eric Johnston, Konrad Kieling, Naomi Nickerson, Mihir Pant, Fernando Pastawski, Terry Rudolph, and Chris Sparrow.
\newblock ``Fusion-based quantum computation''~(2021).
\newblock  \href{http://arxiv.org/abs/2101.09310}{arXiv:2101.09310}.

\bibitem{lim_repeat-until-success_2005}
Yuan~Liang Lim, Almut Beige, and Leong~Chuan Kwek.
\newblock ``Repeat-{Until}-{Success} {Linear} {Optics} {Distributed} {Quantum} {Computing}''.
\newblock \href{https://dx.doi.org/10.1103/PhysRevLett.95.030505}{Physical Review Letters {\bf 95}, 030505}~(2005).

\bibitem{gliniasty_spin-optical_2024}
Grégoire~de Gliniasty, Paul Hilaire, Pierre-Emmanuel Emeriau, Stephen~C. Wein, Alexia Salavrakos, and Shane Mansfield.
\newblock ``A {Spin}-{Optical} {Quantum} {Computing} {Architecture}''~(2024).
\newblock arXiv:2311.05605.

\bibitem{lee_nearly_2015}
Seung-Woo Lee, Kimin Park, Timothy~C. Ralph, and Hyunseok Jeong.
\newblock ``Nearly deterministic {Bell} measurement with multiphoton entanglement for efficient quantum-information processing''.
\newblock \href{https://dx.doi.org/10.1103/PhysRevA.92.052324}{Physical Review A {\bf 92}, 052324}~(2015).

\bibitem{hilaire_enhanced_2024}
Paul Hilaire, Théo Dessertaine, Boris Bourdoncle, Aurélie Denys, Grégoire~de Gliniasty, Gerard Valentí-Rojas, and Shane Mansfield.
\newblock ``Enhanced {Fault}-tolerance in {Photonic} {Quantum} {Computing}: {Floquet} {Code} {Outperforms} {Surface} {Code} in {Tailored} {Architecture}''~(2024).
\newblock arXiv:2410.07065 [quant-ph].

\bibitem{felice_fusion_2024}
Giovanni de~Felice, Boldizsár Poór, Lia Yeh, and William Cashman.
\newblock ``Fusion and flow: formal protocols to reliably build photonic graph states''~(2024).
\newblock arXiv:2409.13541.

\bibitem{duncan_rewriting_2010}
Ross Duncan and Simon Perdrix.
\newblock ``Rewriting {Measurement}-{Based} {Quantum} {Computations} with {Generalised} {Flow}''.
\newblock In David Hutchison, Takeo Kanade, Josef Kittler, Jon~M. Kleinberg, Friedemann Mattern, John~C. Mitchell, Moni Naor, Oscar Nierstrasz, C.~Pandu~Rangan, Bernhard Steffen, Madhu Sudan, Demetri Terzopoulos, Doug Tygar, Moshe~Y. Vardi, Gerhard Weikum, Samson Abramsky, Cyril Gavoille, Claude Kirchner, Friedhelm Meyer auf~der Heide, and Paul~G. Spirakis, editors, Automata, {Languages} and {Programming}.
\newblock \href{https://dx.doi.org/10.1007/978-3-642-14162-1_24}{Volume 6199, pages 285--296}.
\newblock Springer Berlin Heidelberg, Berlin, Heidelberg~(2010).

\bibitem{duncan_graph-theoretic_2019}
Ross Duncan, Aleks Kissinger, Simon Pedrix, and John van~de Wetering.
\newblock ``Graph-theoretic {Simplification} of {Quantum} {Circuits} with the {ZX}-calculus''~(2019).
\newblock  url:~\url{http://arxiv.org/abs/1902.03178}.

\bibitem{browne_generalized_2007}
Daniel~E. Browne, Elham Kashefi, Mehdi Mhalla, and Simon Perdrix.
\newblock ``Generalized flow and determinism in measurement-based quantum computation''.
\newblock \href{https://dx.doi.org/10.1088/1367-2630/9/8/250}{New Journal of Physics {\bf 9}, 250}~(2007).

\bibitem{backens_there_2021}
Miriam Backens, Hector Miller-Bakewell, Giovanni~de Felice, Leo Lobski, and John van~de Wetering.
\newblock ``There and back again: {A} circuit extraction tale''.
\newblock \href{https://dx.doi.org/10.22331/q-2021-03-25-421}{Quantum {\bf 5}, 421}~(2021).

\bibitem{pauli_preserving_rewrites}
Tommy McElvanney and Miriam Backens.
\newblock ``Complete flow-preserving rewrite rules for mbqc patterns with pauli measurements''.
\newblock \href{https://dx.doi.org/10.4204/eptcs.394.5}{Electronic Proceedings in Theoretical Computer Science {\bf 394}, 66–82}~(2023).

\bibitem{cite_qec_benchmarks}
Seok-Hyung Lee and Hyunseok Jeong.
\newblock ``Graph-theoretical optimization of fusion-based graph state generation''.
\newblock \href{https://dx.doi.org/10.22331/q-2023-12-20-1212}{Quantum {\bf 7}, 1212}~(2023).

\bibitem{cite_qasm_bench}
Ang Li, Samuel Stein, Sriram Krishnamoorthy, and James Ang.
\newblock ``Qasmbench: A low-level qasm benchmark suite for nisq evaluation and simulation''~(2022).
\newblock  \href{http://arxiv.org/abs/2005.13018}{arXiv:2005.13018}.

\bibitem{wein_minimizing_2024}
Stephen~C. Wein, Timothée~Goubault de~Brugière, Luka Music, Pascale Senellart, Boris Bourdoncle, and Shane Mansfield.
\newblock ``Minimizing resource overhead in fusion-based quantum computation using hybrid spin-photon devices''~(2024).
\newblock  url:~\url{https://arxiv.org/abs/2412.08611v1}.

\bibitem{cite_vienna_team}
Felix Zilk, Korbinian Staudacher, Tobias Guggemos, Karl Fürlinger, Dieter Kranzlmüller, and Philip Walther.
\newblock ``A compiler for universal photonic quantum computers''.
\newblock In 2022 IEEE/ACM Third International Workshop on Quantum Computing Software (QCS).
\newblock IEEE~(2022).

\bibitem{cite_rus_original}
Yuan~Liang Lim, Almut Beige, and Leong~Chuan Kwek.
\newblock ``Repeat-until-success linear optics distributed quantum computing''.
\newblock \href{https://dx.doi.org/10.1103/physrevlett.95.030505}{Physical Review Letters{\bf 95}}~(2005).

\bibitem{cite_spdc_ghz}
Yun-Feng Huang, Bi-Heng Liu, Liang Peng, Yu-Hu Li, Li~Li, Chuan-Feng Li, and Guang-Can Guo.
\newblock ``Experimental generation of an eight-photon greenberger--horne--zeilinger state''.
\newblock \href{https://dx.doi.org/10.1038/ncomms1556}{Nature Communications {\bf 2}, 546}~(2011).

\bibitem{cite_deterministic_solid_state_emitter}
H.~Huet, P.~R. Ramesh, S.~C. Wein, N.~Coste, P.~Hilaire, N.~Somaschi, M.~Morassi, A.~Lema{\^\i}tre, I.~Sagnes, M.~F. Doty, O.~Krebs, L.~Lanco, D.~A. Fioretto, and P.~Senellart.
\newblock ``Deterministic and reconfigurable graph state generation with a single solid-state quantum emitter''.
\newblock \href{https://dx.doi.org/10.1038/s41467-025-59693-3}{Nature Communications {\bf 16}, 4337}~(2025).

\bibitem{cite_matter_based_linear_cluster_state_generation_2}
D.~Istrati, Y.~Pilnyak, J.~C. Loredo, C.~Ant{\'o}n, N.~Somaschi, P.~Hilaire, H.~Ollivier, M.~Esmann, L.~Cohen, L.~Vidro, C.~Millet, A.~Lema{\^\i}tre, I.~Sagnes, A.~Harouri, L.~Lanco, P.~Senellart, and H.~S. Eisenberg.
\newblock ``Sequential generation of linear cluster states from a single photon emitter''.
\newblock \href{https://dx.doi.org/10.1038/s41467-020-19341-4}{Nature Communications {\bf 11}, 5501}~(2020).

\bibitem{cite_deterministic_catapillar_states}
H.~Huet, P.~R. Ramesh, S.~C. Wein, N.~Coste, P.~Hilaire, N.~Somaschi, M.~Morassi, A.~Lemaître, I.~Sagnes, M.~F. Doty, O.~Krebs, L.~Lanco, D.~A. Fioretto, and P.~Senellart.
\newblock ``Deterministic and reconfigurable graph state generation with a single solid-state quantum emitter''~(2025).
\newblock  \href{http://arxiv.org/abs/2410.23518}{arXiv:2410.23518}.

\bibitem{cite_karp_21}
Richard~M. Karp.
\newblock ``Reducibility among combinatorial problems''.
\newblock \href{https://dx.doi.org/10.1007/978-1-4684-2001-2_9}{Pages 85--103}.
\newblock Springer US. Boston, MA~(1972).

\bibitem{cite_catalog_hamiltonian}
``List of graphs for which the hamiltonian path problem is in p''.
\newblock  url:~\url{https://www.graphclasses.org/classes/problem_Hamiltonian_path.html}.

\bibitem{cite_approx_mpc}
Shlomo Moran, Ilan Newman, and Yaron Wolfsthal.
\newblock ``Approximation algorithms for covering a graph by vertex-disjoint paths of maximum total weight''.
\newblock Networks {\bf 20}, 55--64~(1990).
\newblock  url:~\url{https://api.semanticscholar.org/CorpusID:8146642}.

\bibitem{cite_mpc}
Kenya Kobayashi, Guohui Lin, Eiji Miyano, Toshiki Saitoh, Akira Suzuki, Tadatoshi Utashima, and Tsuyoshi Yagita.
\newblock ``Path cover problems with length cost''.
\newblock \href{https://dx.doi.org/10.1007/s00453-023-01106-2}{Algorithmica {\bf 85}, 3348--3375}~(2023).

\bibitem{cite_eulers_theorem}
Norman~L. Biggs, E.~Keith Lloyd, and Robin~J. Wilson.
\newblock ``Graph theory 1736-1936''.
\newblock Clarendon Press. ~(1976).
\newblock  url:~\url{https://api.semanticscholar.org/CorpusID:118963253}.

\bibitem{cite_orig_proof_of_mtd}
Oystein Ore and Robin~J. Wilson.
\newblock ``Graphs and their uses''.
\newblock \href{https://dx.doi.org/https://doi.org/10.5948/UPO9780883859490}{Anneli Lax New Mathematical Library}. Mathematical Association of America. ~(1990).

\bibitem{cite_hierholzer_alg}
Herbert Fleischner.
\newblock ``Eulerian graphs and related topics: Part 1, volume 2''.
\newblock Annals of Discrete Mathematics. ~(1991).
\newblock  url:~\url{https://api.semanticscholar.org/CorpusID:118183786}.

\bibitem{cite_claw_free_matching}
David~P. Sumner.
\newblock ``Graphs with 1-factors''.
\newblock Proceedings of the American Mathematical Society {\bf 42}, 8--12~(1974).
\newblock  url:~\url{http://www.jstor.org/stable/2039666}.

\bibitem{cite_poly_matching}
Silvio Micali and Vijay~V. Vazirani.
\newblock ``An o(v|v| c |e|) algoithm for finding maximum matching in general graphs''.
\newblock In 21st Annual Symposium on Foundations of Computer Science (sfcs 1980).
\newblock \href{https://dx.doi.org/10.1109/SFCS.1980.12}{Pages 17--27}.
\newblock ~(1980).

\bibitem{cite_bin_pack_approx}
Vijay~V. Vazirani.
\newblock ``Approximation algorithms''.
\newblock \href{https://dx.doi.org/https://doi.org/10.1007/978-3-662-04565-7}{Page 380}.
\newblock Springer Berlin, Heidelberg. ~(2001).

\bibitem{cite_cubic_hamiltonian_nphard}
M.~R. Garey, D.~S. Johnson, and R.~Endre Tarjan.
\newblock ``The planar hamiltonian circuit problem is np-complete''.
\newblock \href{https://dx.doi.org/10.1137/0205049}{SIAM Journal on Computing {\bf 5}, 704--714}~(1976).
\newblock  \href{http://arxiv.org/abs/https://doi.org/10.1137/0205049}{arXiv:https://doi.org/10.1137/0205049}.

\bibitem{cite_tsp_benchmarks}
Pouya Baniasadi, Vladimir Ejov, Michael Haythorpe, and Serguei Rossomakhine.
\newblock ``A new benchmark set for traveling salesman problem and hamiltonian cycle problem''~(2018).
\newblock  \href{http://arxiv.org/abs/1806.09285}{arXiv:1806.09285}.

\bibitem{cite_cz_complexity}
Soh Kumabe, Ryuhei Mori, and Yusei Yoshimura.
\newblock ``Complexity of graph-state preparation by clifford circuits''~(2024).
\newblock  \href{http://arxiv.org/abs/2402.05874}{arXiv:2402.05874}.

\bibitem{cite_rank_width}
Sang il~Oum.
\newblock ``Rank-width and vertex-minors''.
\newblock \href{https://dx.doi.org/https://doi.org/10.1016/j.jctb.2005.03.003}{Journal of Combinatorial Theory, Series B {\bf 95}, 79--100}~(2005).

\bibitem{cite_orig_fbqc}
Sara Bartolucci, Patrick Birchall, Hector Bombin, Hugo Cable, Chris Dawson, Mercedes Gimeno-Segovia, Eric Johnston, Konrad Kieling, Naomi Nickerson, Mihir Pant, Fernando Pastawski, Terry Rudolph, and Chris Sparrow.
\newblock ``Fusion-based quantum computation''~(2021).
\newblock  \href{http://arxiv.org/abs/2101.09310}{arXiv:2101.09310}.

\bibitem{cite_lattice_mbqc}
Robert Raussendorf, Daniel~E. Browne, and Hans~J. Briegel.
\newblock ``Measurement-based quantum computation on cluster states''.
\newblock \href{https://dx.doi.org/10.1103/PhysRevA.68.022312}{Phys. Rev. A {\bf 68}, 022312}~(2003).

\bibitem{cite_repeaters}
Koji Azuma, Kiyoshi Tamaki, and Hoi-Kwong Lo.
\newblock ``All-photonic quantum repeaters''.
\newblock \href{https://dx.doi.org/10.1038/ncomms7787}{Nature Communications {\bf 6}, 6787}~(2015).

\bibitem{cite_trees_mbqc}
Ying Li, Peter~C. Humphreys, Gabriel~J. Mendoza, and Simon~C. Benjamin.
\newblock ``Resource costs for fault-tolerant linear optical quantum computing''.
\newblock \href{https://dx.doi.org/10.1103/PhysRevX.5.041007}{Phys. Rev. X {\bf 5}, 041007}~(2015).

\bibitem{cite_loss_tolerance_trees}
Michael Varnava, Daniel~E. Browne, and Terry Rudolph.
\newblock ``Loss tolerance in one-way quantum computation via counterfactual error correction''.
\newblock \href{https://dx.doi.org/10.1103/PhysRevLett.97.120501}{Phys. Rev. Lett. {\bf 97}, 120501}~(2006).

\bibitem{cite_pyzx}
Aleks Kissinger and John van~de Wetering.
\newblock ``Pyzx: Large scale automated diagrammatic reasoning''.
\newblock \href{https://dx.doi.org/10.4204/eptcs.318.14}{Electronic Proceedings in Theoretical Computer Science {\bf 318}, 229–241}~(2020).

\bibitem{cite_linear_cluster_state_generation}
D.~Istrati, Y.~Pilnyak, J.~C. Loredo, C.~Ant{\'o}n, N.~Somaschi, P.~Hilaire, H.~Ollivier, M.~Esmann, L.~Cohen, L.~Vidro, C.~Millet, A.~Lema{\^i}tre, I.~Sagnes, A.~Harouri, L.~Lanco, P.~Senellart, and H.~S. Eisenberg.
\newblock ``Sequential generation of linear cluster states from a single photon emitter''.
\newblock \href{https://dx.doi.org/https://doi.org/10.1038/s41467-020-19341-4}{Nature Communications {\bf 11}, 5501}~(2020).

\bibitem{lobl_generating_2025}
Matthias~C. Löbl, Love~A. Pettersson, Andrew Jena, Luca Dellantonio, Stefano Paesani, and Anders~S. Sørensen.
\newblock ``Generating graph states with a single quantum emitter and the minimum number of fusions''~(2025).
\newblock arXiv:2412.04587 [quant-ph].

\bibitem{bell_optimizing_2023}
Thomas~J. Bell, Love~A. Pettersson, and Stefano Paesani.
\newblock ``Optimizing {Graph} {Codes} for {Measurement}-{Based} {Loss} {Tolerance}''.
\newblock \href{https://dx.doi.org/10.1103/PRXQuantum.4.020328}{PRX Quantum {\bf 4}, 020328}~(2023).

\bibitem{lobl_loss-tolerant_2024}
Matthias~C. Löbl, Stefano Paesani, and Anders~S. Sørensen.
\newblock ``Loss-tolerant architecture for quantum computing with quantum emitters''.
\newblock \href{https://dx.doi.org/10.22331/q-2024-03-28-1302}{Quantum {\bf 8}, 1302}~(2024).

\end{thebibliography}
